\documentclass[prl,twocolumn,showpacs]{revtex4-2}

\usepackage{amsmath}
\usepackage{latexsym}
\usepackage{amssymb}
\usepackage{graphicx}
\usepackage[colorlinks=true, citecolor=blue, urlcolor=blue]{hyperref}
\usepackage{float}
\usepackage{amsfonts}
\usepackage{textcomp}
\usepackage{mathpazo}
\usepackage{comment}
\usepackage{xr}

\usepackage{bbm}

\usepackage{xcolor}
\definecolor{myurlcolor}{rgb}{0,0,0.4}
\definecolor{mycitecolor}{rgb}{0,0.5,0}
\definecolor{myrefcolor}{rgb}{0.5,0,0}
\usepackage{hyperref}
\hypersetup{colorlinks,
linkcolor=myrefcolor,
citecolor=mycitecolor,
urlcolor=myurlcolor}

\usepackage[draft]{fixme}
\usepackage{amsmath,bbm}
\usepackage{graphicx}
\usepackage{amsfonts}
\usepackage{amssymb}
\usepackage{amsmath, amssymb, amsthm,verbatim,graphicx,bbm}
\usepackage{mathrsfs}
\usepackage{color,xcolor,longtable}

\usepackage{xr}
\makeatletter
\newcommand*{\addFileDependency}[1]{
  \typeout{(#1)}
  \@addtofilelist{#1}
  \IfFileExists{#1}{}{\typeout{No file #1.}}
}
\makeatother

\newcommand*{\myexternaldocument}[1]{
    \externaldocument{#1}
    \addFileDependency{#1.tex}
    \addFileDependency{#1.aux}
}

\myexternaldocument{manuscript_PRL}

\newcommand{\beq}[0]{\begin{equation}}
\newcommand{\eeq}[0]{\end{equation}}

\newcommand{\one}{\leavevmode\hbox{\small1\normalsize\kern-.33em1}}

\def\be{\begin{equation}}
\def\ee{\end{equation}}
\def\ben{\begin{eqnarray}}
\def\een{\end{eqnarray}}
\def\eea{\end{array}}
\def\bea{\begin{array}}

\newcommand{\Tr}[1]{\mathrm{Tr}#1}
\newcommand{\bei}{\begin{itemize}}
\newcommand{\eei}{\end{itemize}}
\newcommand{\ket}[1]{|#1\rangle}
\newcommand{\bra}[1]{\langle#1|}

\newcommand{\proj}[1]{\ket{#1}\!\!\bra{#1}}

\newcommand{\I}{\mathbbm{1}}

\renewcommand{\emph}[1]{\textbf{#1}}

\makeatletter
\newtheorem*{rep@theorem}{\rep@title}
\newcommand{\newreptheorem}[2]{%
\newenvironment{rep#1}[1]{%
 \def\rep@title{#2 \ref{##1}}%
 \begin{rep@theorem}}%
 {\end{rep@theorem}}}
\makeatother

\theoremstyle{plain}
\newtheorem{thm}{Theorem}
\newtheorem*{thm*}{Theorem}
\newreptheorem{thm}{Theorem}

\newtheorem{fakt}{Fact}

\theoremstyle{definition}

\theoremstyle{remark}

\usepackage[T1]{fontenc}

\begin{document}

\title{Model-independent inference of quantum interaction from statistics}
\author{Shubhayan Sarkar}
\email{shubhayan.sarkar@ulb.be}
\affiliation{Laboratoire d’Information Quantique, Université libre de Bruxelles (ULB), Av. F. D. Roosevelt 50, 1050 Bruxelles, Belgium}

\begin{abstract}	
Any physical theory aims to establish the relationship between physical systems in terms of the interaction between these systems. However, any known approach in the literature to infer this interaction is dependent on the particular modelling of the physical systems involved. Here, we propose an alternative approach where one does not need to model the systems involved but only assume that these systems behave according to quantum theory. We first propose a setup to infer a particular entangling quantum interaction between two systems from the statistics. For our purpose, we utilise the framework of Bell inequalities. We then extend this setup where an arbitrary number of quantum systems interact via some entangling interaction.
\end{abstract}


\maketitle

{\textit{Introduction---}} In physics, the exploration of interactions between two systems constitutes the foundational aspect of understanding natural phenomena. Despite significant advancements in the field, the analyses conducted thus far remain inherently model-dependent. A notable example of this lies in quantum field theories, where one hypothesizes different fields to explain experimental observations. While these models have proven immensely powerful in explaining a wide array of physical phenomena, their reliance on specific theoretical constructs underscores the necessity for continued refinement and exploration. Furthermore, these models are highly dependent on the parameters and the assumptions made on the experimental setup.



Recently, the idea of device-independent (DI) certification of quantum states and measurements has gained a lot of interest as they allow one to certify the properties of an unknown quantum device by only observing the statistical data it generates and making minimal assumptions about the device. The essential resource for any DI scheme is Bell nonlocality \cite{Bell,Bell66,NonlocalityReview}. For instance, any violation of a Bell inequality is a DI certification of the presence of entanglement inside the device.

The strongest form of DI certification is termed self-testing \cite{Mayers_selftesting, Yao}, enabling near-complete characterization of the underlying quantum state and its associated measurements by only assuming that the devices behave according to quantum theory.
%
%
Consequently, a wide range of schemes has been proposed to self-test pure entangled quantum states and projective quantum measurements (see, e.g., Refs. \cite{Scarani,Reichardt_nature,Mckague_2014,Wu_2014,Bamps,All,chainedBell, Projection,Jed1,prakash,Armin1,Flavio, sarkar,sarkaro2,Marco, Allst,sarkar2023universal}) as well as mixed entangled states  \cite{sarkar2023self, sarkar2023universal} and non-projective measurements \cite{remik1,sarkar2023universal}. Furthermore, schemes to self-test single unitaries \cite{Dall} and the controlled-not gate have been proposed in \cite{pavel}. Despite this progress, no scheme has been proposed that can be used to certify the interaction between two unknown systems. 

Inspired by self-testing, in this work, we propose a model-independent approach to infer the quantum interaction between two systems from the statistics generated in the experiment. We do not delve into the physical considerations of the experimental setup but rather focus on the operational nature of the scheme, that is, we do not care about the degree of freedom in which the interaction between the systems takes place. However, the concerned degree of freedom is the one that is being measured by the detectors. We particularly focus on entangling quantum interactions, that is, quantum interactions that can generate an entangled state from a product state. 

Recently, a lot of attention has been devoted to such interactions as they can be a certificate to probe the quantum nature of gravity \cite{Vedral,Bose}. Such entangling interactions have also been explored in quantum electrodynamics, for instance, two electrons that dynamically scatter get entangled either in spin or momentum degrees of freedom  \cite{eeint1,eeint2,eeint3,eeint4,eeint5,eeint6,eeint7,eeint8,eeint9}. Furthermore, such entangling interactions have also been explored in the quark-quark system  \cite{LHC1, LHC2, LHC3, LHC4} and has been recently observed at the Large Hadron Collider \cite{LHC}. We first propose a scenario that can be used to infer that two systems are interacting via a particular entangling quantum interaction that can generate maximally entangled states from product ones. For our purpose, we use the Bell inequalities suggested in \cite{Marco, sarkar2023universal}. Then, we generalize this result where an arbitrary number of systems interact with each other. Again, we utilize the Bell inequalities suggested in \cite{sarkar2023universal} to certify a particular entangling quantum interaction that can create Greenberger-Horne-Zeilinger(GHZ)-like states from product states.

{\textit{Two-system interaction---}} Let us begin by describing the setup to infer the entangling quantum interaction in a model-independent way.
\begin{figure*}[t!]
    \centering
    \includegraphics[scale=.4]{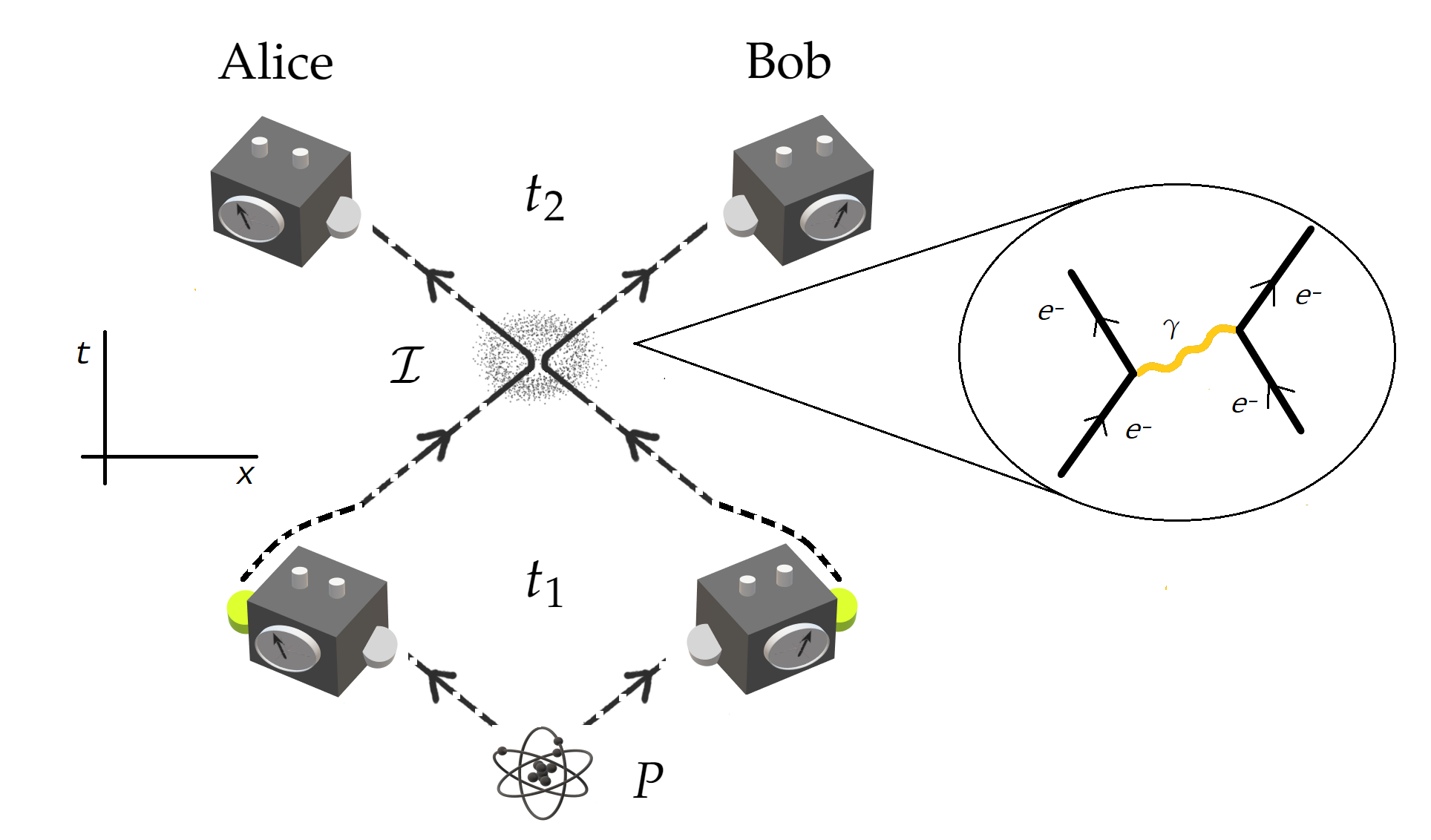}
    \caption{(Left) Setup for model-independent inference of quantum interaction. Particles from source $P$ are sent to spatially separated labs where Alice and Bob reside. At time $t_1$, Alice and Bob choose $x_1, y_1 = 0,1$ for local measurements, yielding $a_1, b_1 = 0,1$. Post-measured particles interact via $\mathcal{I}$, and return to their respective labs at $t_2$. At $t_2$, inputs $x_2, y_2 = 0,1$ are chosen by Alice and Bob respectively resulting in $a_2, b_2 = 0,1$. (Right) An example of an entangling quantum interaction. Two electrons interacting via dynamical scattering can get entangled either in spin or momentum degrees of freedom.}
    \label{fig1}
\end{figure*} 

A source $P$ sends particles to Alice and Bob who are located in spatially separated labs. Alice and Bob can freely choose two inputs each denoted by $x_1,y_1=0,1$ respectively based on which local measurements are performed on their particles at time $t_1$. Furthermore, each measurement results in two outputs denoted by $a_1,b_1=0,1$ for Alice and Bob respectively. Then, the post-measured particles are allowed to leave their labs after which they interact via some interaction $\mathcal{I}$. After the interaction, the particles come back to Alice and Bob at time $t_2$. Again, Alice and Bob choose two inputs $x_2,y_2=0,1$ based on which the incoming particles are measured and results in two outputs $a_2,b_2=0,1$ [see Fig. \ref{fig1}]. 

By repeating the above-described procedure, they obtain two joint probability distributions often referred to as correlations. The first one is obtained at time $t_1$ given by $\vec{p}_1=\{p(a_1,b_1|x_1,y_1,P)\}$ where $p(a_1,b_1|x_1,y_1,P)$ denotes the probability of obtaining outcomes $a_1,b_1$ given the inputs $x_1,y_1$ given the source $P$. The second probability distribution is obtained at time $t_2$ given by 
\begin{eqnarray}
    \vec{p}_{2,a_1,b_1,x_1,y_1}=\{p(a_2,b_2|x_2,y_2,a_1,b_1,x_1,y_1,\mathcal{I},P)\}
\end{eqnarray}
where $p(a_2,b_2|x_2,y_2,a_1,b_1,x_1,y_1,\mathcal{I},P)$ denotes the conditional probability of obtaining outcome $a_2,b_2$ at time $t_2$ given inputs $x_2,y_2$ when Alice and Bob at time $t_1$ obtained $a_1,b_1$ with inputs $x_1,y_1$ and the post-measured states interacted via some interaction $\mathcal{I}$. 

Let us now describe the above-proposed scenario [see Fig \ref{fig1}] within quantum theory. The probability $p(a_1,b_1|x_1,y_1,P)$ is obtained via the Born rule as
\begin{eqnarray}\label{probs}
    p(a_1,b_1|x_1,y_1,P)=\Tr(M^A_{a_1|x_1}\otimes M^B_{b_1|y_1}\rho_{AB})
\end{eqnarray}
where $M^A_{a_1|x_1}, M^B_{b_1|y_1}$ denote the measurement elements of Alice and Bob such that these elements are positive and $\sum_aM^A_{a_1|x_1}=\sum_bM^B_{b_1|y_1}=\I$. Here, $\rho_{AB}$ is the quantum state prepared by the source $P$. One could similarly define the above formula \eqref{probs} at time $t_2$ where the state $\rho_{AB}$ will be replaced by the post-interaction states defined below in \eqref{post-int statem}. It is often helpful to express the correlations in terms of the expected values of observables which are defined as
\begin{equation}\label{obspic}
    \langle A_{m}\otimes B_{l} \rangle
     = \sum_{a,b=0}^{1} (-1)^{a+b} p(a,b|m,l,\ldots,P)
\end{equation}
Notice that by using Eq. \eqref{probs}, these expectation values can be expressed as $\langle A_{m}\otimes B_{l} \rangle = \Tr[A_{m}\otimes B_{l}\rho_{AB}]$, where $A_{m},B_{l}$ are quantum operators, referred to as observables, defined via the measurement elements as $A_{m}=M^A_{0|m}-M^A_{1|m},\ B_{l}=M^B_{0|l}-M^B_{1|l}$ for every $m,l$. When the measurement is projective, then the corresponding observable is unitary. Furthermore, if the measurements of Alice and Bob are projective and they observe the outputs $a_1,b_1$ given the inputs $x_1,y_1$ respectively, then the post-measurement states are given by
\begin{eqnarray}
    \rho'_{a_1,b_1,x_1,y_1}=\frac{M^A_{a_1|x_1}\otimes M^B_{b_1|y_1}\rho_{AB}M^A_{a_1|x_1}\otimes M^B_{b_1|y_1}}{p(a_1,b_1|x_1,y_1)}.
\end{eqnarray}
Then, the particles interact via some Hamiltonian $H(t)$ for a time $\delta t$. Consequently, the state $\rho'_{a_1,b_1,x_1,y_1}$ at time $t_1$ evolves via the unitary process $V(\delta t)$ \cite{Sakurai}. For instance, when the Hamiltonian is time-independent $H(t)\equiv H$ for any time $t$, then $V(\delta t)$ is given by $ V(\delta t)=e^{-iH\delta t}$. Now, the states after the interaction, referred to as post-interaction states, is given by 
\begin{eqnarray}\label{post-int statem}
    V(\delta t)\rho'_{a_1,b_1,x_1,y_1}V(\delta t)^{\dagger}=\sigma_{a_1,b_1,x_1,y_1}.
\end{eqnarray}
For simplicity, we will further represent $V(\delta t)\equiv V$. Without loss of generality, we  consider that the unitary maps quantum states from the Hilbert space $\mathcal{H}_{A(t_1)}\otimes\mathcal{H}_{B(t_1)}$ to a different Hilbert space $\mathcal{H}_{A(t_2)}\otimes\mathcal{H}_{B(t_2)}$.

\textit{Model-independent inference---} Inspired by the idea of self-testing [see \cite{supic1} for a review], let us introduce the idea of model-independent inference of a quantum interaction via statistics by referring back to the scenario shown in Fig. \ref{fig1}. First, we consider that the measurements conducted by the parties, the state prepared by the source as well as the interaction between the post-measured states are unknown except for the fact that they obey quantum theory. The only other information that Alice and Bob
have about the whole scenario is via the observed correlations $\vec{p}_1, \vec{p}_{2,a_1,b_1,x_1,y_1}$.
It is worth mentioning here that we assume that dimensions of the local Hilbert spaces $\mathcal{H}_{A(t_i)},\mathcal{H}_{B(t_i)}$ for $i=1,2$ is unknown but finite. 

Let us then consider a reference experiment giving rise to the same correlations $\vec{p}_1$ when some known observables $A'_{m},B'_{l}$ are performed on a known quantum state prepared by the source 
$\ket{\psi'_{AB}}\in\mathcal{H}_{A'(t_1)}\otimes\mathcal{H}_{B'(t_1)}$. Then the post-measured states interact via some known unitary $V'$ after which they are measured using the same known observables $A'_{m},B'_{l}$ to obtain $\vec{p}_{2,a_1,b_1,x_1,y_1}$. The task of model-independent inference is to deduce from the observed $\vec{p}_1,\vec{p}_{2,a_1,b_1,x_1,y_1}$ that the actual experiment is equivalent to the reference one in the following sense: (i) the local Hilbert spaces admit the product form $\mathcal{H}_{A(t_i)}=\mathcal{H}_{A'(t_i)}\otimes\mathcal{H}_{A''(t_i)}$ and $\mathcal{H}_{B(t_i)}=\mathcal{H}_{B'(t_i)}\otimes\mathcal{H}_{B''(t_i)}$ for some auxiliary Hilbert spaces $\mathcal{H}_{A''(t_i)}$ and $\mathcal{H}_{B''(t_i)}$; (ii) there are local unitary operations  $U_{s(t_i)}:\mathcal{H}_{s(t_i)}\to \mathcal{H}_{s'(t_i)}\otimes\mathcal{H}_{s''(t_i)}$ for $s=A,B$ and $i=1,2$
such that
\begin{eqnarray}\label{ststate}
 (U_{A(t_1)}\otimes U_{B(t_1)})\rho_{AB}(U_{A(t_1)}\otimes U_{B(t_1)})^{\dagger}\qquad\nonumber\\ \qquad=\proj{\psi'}_{A'(t_1)B'(t_1)}\otimes\xi_{A''(t_1)B''(t_1)},
 %
\end{eqnarray}
where  $\xi_{A_i''E_i''}$ acts on $\mathcal{H}_{A''(t_i)}\otimes\mathcal{H}_{B''(t_i)}$ is some auxiliary quantum state, and
\begin{eqnarray}\label{stmea}
U_{A(t_i)}\,\overline{A}_{m(t_i)}\,U_{A(t_i)}^{\dagger}&=&A_{m}'\otimes\mathbbm{1}_{A''(t_i)},\nonumber \\U_{B(t_i)}\,\overline{B}_{l(t_i)}\,U_{B(t_i)}^{\dagger}&=&B_{l}'\otimes\mathbbm{1}_{B''(t_i)},
\end{eqnarray}
where $\mathbbm{1}_{s''(t_i)}$ is the identity acting on the parties auxiliary system and $\overline{A}_{m(t_i)}=\Pi^A_{i}A_{m}\Pi^A_{i}$ and  $\overline{B}_{l(t_i)}=\Pi^B_{i}B_{l}\Pi^B_{i}$ such that $\Pi^A_{i},\Pi^B_{i}$ denotes the projection of the observables $A_{m},B_{l}$ onto the Hilbert space $\mathcal{H}_{A(t_i)},\mathcal{H}_{B(t_i)}$ respectively; (iii) The interaction $V$ is certified as 
\begin{eqnarray}\label{stint}
    (U_{A(t_2)}\otimes U_{B(t_2)})\ V\ (U_{A(t_1)}\otimes U_{B(t_1)})^{\dagger}=V'\otimes V_0
\end{eqnarray}
where $V_0$ is a unitary matrix mapping $\mathcal{H}_{A''(t_1)}\otimes\mathcal{H}_{B''(t_1)}$ to $\mathcal{H}_{A''(t_2)}\otimes\mathcal{H}_{B''(t_2)}$.
For a note, if the above conditions (i) and (ii) are met one says that 
the reference state and measurements are self-tested in the actual experiment from the observed correlations. Then if condition (iii) is met, the interaction between the systems is certified in a model-independent way. 

Consider now  the following Bell inequalities for $a_1,b_1=0,1$
\begin{eqnarray}\label{BE1Nm}
\mathcal{B}_{a_1,b_1}=(-1)^{a_1} \left\langle\tilde{A}_{1}\otimes B_{1} +(-1)^{b_1}\tilde{A}_{0}\otimes B_{0} \right\rangle\leqslant  \beta_C
\end{eqnarray}
where,
\begin{eqnarray}\label{overAm}
    \tilde{A}_{0}=\frac{A_{0}-A_{1}}{\sqrt{2}},\qquad\tilde{A}_{1}=\frac{A_{0}+A_{1}}{\sqrt{2}}.
\end{eqnarray}
The classical bound of the above Bell inequalities is $\beta_C=\sqrt{2}$ for any $a_1,b_1$. 
Consider now the following states 
\begin{equation}\label{GHZvecsm}
\ket{\phi_{a_1,b_1}}=\frac{1}{\sqrt{2}}(\ket{a_1 b_1}+(-1)^{a_{1}}|a_1^{\perp}b_1^{\perp}\rangle),
\end{equation}
where $a_1^{\perp}=1-a_1,b_1^{\perp}=1-b_1$, and the following observables
\begin{equation}\label{GHZObsm}
A_{0}=\frac{X+Z}{\sqrt{2}},\quad A_{1}= \frac{X-Z}{\sqrt{2}}, \quad B_{0}=Z,\quad B_{1}=X. 
\end{equation}
As shown in Fact 1 of Appendix A, using these states \eqref{GHZvecsm} and observables \eqref{GHZObsm}, one can attain the value $\mathcal{B}_{a_1,b_1}=2$. This is the quantum bound $\beta_Q$ of $\mathcal{B}_{a_1,b_1}$, that is, the maximal value of $\mathcal{B}_{a_1,b_1}$ that can be attained within quantum theory.

Let us now suppose that correlations $\vec{p}_1$ achieves the quantum bound of $\mathcal{B}_{0,0}$, that is, the quantum state $\rho_{AB}$ maximally violate the Bell inequality $\mathcal{B}_{0,0}$. Furthermore, the correlations $\vec{p}_{2,a_1,b_1,0,0}$ achieves the quantum bound of $\mathcal{B}_{a_1,b_1}$ for each $a_1,b_1$, that is, the post-interaction quantum states $\sigma_{a_1,b_1,0,0}$ maximally violate the Bell inequalities $\mathcal{B}_{a_1,b_1}$. Along with them, one also needs to observe when $a_1=b_1=0$ and $x_1=1,y_1=1$, the post-interaction states $\sigma_{0,0,1,1}$ satisfy 
\begin{eqnarray}\label{extrastatm}
    \langle\I\otimes B_{1}\rangle= -\langle \tilde{A}_{0}\otimes\I\rangle=1.
\end{eqnarray}

Consider now the following unitary
\begin{eqnarray}\label{U2m}
\mathcal{U}=\sum_{a_1,b_1=0,1}\ket{\phi_{a_1,b_1}}\!\bra{\overline{a}_1b_1}
\end{eqnarray}
where, the $\{\ket{\overline{0}},\ket{\overline{1}}\}$ are the eigenvectors of $(X+Z)/\sqrt{2}$. It is straightforward to verify that if the state $\ket{\phi_{0,0}}$ \eqref{GHZvecsm} after being measured by the observables \eqref{GHZObsm} and obtaining outcome $a_1,b_1$ with inputs $x_1=y_1=0$ evolves via the unitary \eqref{U2m}, then one obtains the post-interaction state as $\ket{\phi_{a_1,b_1}}$. Consequently, all these states and observables satisfy the above-mentioned statistics and thus we take them as the reference quantum states, observables and interaction.

Let us now state the main result.

\begin{thm}\label{theorem1m}
Assume that the Bell inequality $\mathcal{B}_{0,0}$ \eqref{BE1Nm} is maximally violated at $t_1$. Furthermore, when Alice and Bob obtain the outcome $a_1,b_1$ with inputs $x_1=y_1=0$, then the Bell inequalities $\mathcal{B}_{a_1,b_1}$ \eqref{BE1Nm} for any $a_1,b_1$ are maximally violated at $t_2$. Along with it, when Alice and Bob observe the outcomes $a_1=b_1=0$ with inputs $x_1=1,y_1=1$, then the condition \eqref{extrastatm} is also satisfied. Then, the quantum state prepared by the source and the observables of both parties are certified as in \eqref{ststate} and \eqref{stmea}, with the reference strategies given below Eq. \eqref{U2m}. Importantly, the unitary $V$ is certified as defined in \eqref{stint}
\begin{eqnarray}\label{Vm}
    U_{A(t_2)}\otimes U_{B(t_2)}\ V\ U_{A(t_1)}^{\dagger}\otimes U_{B(t_1)}^{\dagger}=\mathcal{U}\otimes V_0
\end{eqnarray}
where $\mathcal{U}$ is given in Eq. \eqref{U2m} and $V_0$ is unitary.
\end{thm}

The proof of the above theorem is presented in Appendix A. Here, we provide a short description of the proof. The proof is mainly divided into two parts. In the first part, we self-test the state prepared by the source and the measurements acting on the local support of this state. This further allows one to certify the post-measurement states of the parties. We then self-test the post-interaction states along with the measurements acting on their local supports. Both of these self-tests are based on the sum of squares (SOS) decomposition of the Bell operator corresponding to the inequality \eqref{BE1Nm} for which we follow the techniques introduced in \cite{sarkar2023universal}. Importantly, we do not assume any state to be pure or measurements to be projective in our proof. In the second part, utilizing the certified post-measured and post-interaction states and the condition \eqref{extrastatm}, we then conclude that the only quantum interaction that can reproduce all the statistics has to be of the form \eqref{Vm}. 

Let us now generalize the above result to the scenario where arbitrary number of particles interact with each other via some entangling quantum interaction.

\textit{N-system interaction---} 
We begin by generalizing the scenario depicted in Fig. \ref{fig1}. A source $P_N$ sends particles to $N$ Alices who are located in spatially separated labs. All the parties freely choose two inputs each denoted by $x_{n,1}=0,1$ for $n=1,\ldots,N$ based on which local measurements are performed on their particles at time $t_1$ respectively. Furthermore, each measurement results in two outputs denoted by $a_{n,1}=0,1$ for the $n$-th Alice. Then, the post-measured particles are allowed to leave their labs after which they interact via some interaction $\mathcal{I}_N$. After the interaction, the particles come back to their respective starting labs at time $t_2$. Again, all the parties choose two inputs $x_{n,2}=0,1$ based on which the incoming particles are measured and results in two outputs $a_{n,2}=0,1$. We further denote $\mathbf{x_i}=x_{1,i},\ldots,x_{N,i}$ and $\mathbf{a_i}=a_{1,i},\ldots,a_{N,i}$ for $i=1,2$. 

By repeating the above-described procedure, they obtain two joint probability distributions or correlations. The first one is obtained at time $t_1$ given by $\vec{p}_{N,1}=\{p(\mathbf{a_1}|\mathbf{x_1},P_N)\}$ where $p(\mathbf{a_1}|\mathbf{x_1},P_N)$ denotes the probability of obtaining outcome $\mathbf{a_1}$ given the input $\mathbf{x_1}$ and source $P_N$. The second probability distribution is obtained at time $t_2$ given by 
\begin{eqnarray}
    \vec{p}_{N,2,\mathbf{a_1},\mathbf{x_1}}=\{p(\mathbf{a_2}|\mathbf{x_2},\mathbf{a_1},\mathbf{x_1},\mathcal{I}_N,P_N)\}
\end{eqnarray}
where $p(\mathbf{a_2}|\mathbf{x_2},\mathbf{a_1},\mathbf{x_1},\mathcal{I}_N,P_N)$ denotes the conditional probability of obtaining outcome $\mathbf{a_2}$ at time $t_2$ given inputs $\mathbf{x_2}$ when Alice and Bob at time $t_1$ obtained $\mathbf{a_1}$ with input $\mathbf{x_1}$ and the post-measured states interacted via some interaction $\mathcal{I}_N$. 

Within quantum theory, the probability $p(\mathbf{a_1}|\mathbf{x_1},P_N)$ is obtained via the Born rule as
\begin{eqnarray}\label{probsN}
    p(\mathbf{a_1}|\mathbf{x_1},P_N)=\Tr\left(\bigotimes_{n=1}^N M^n_{a_{n,1}|x_{n,1}}\rho_{N}\right)
\end{eqnarray}
where $M^n_{a_{n,1}|x_{n,1}}$ denote the measurement elements of $n-$th Alice with $\rho_{N}$ denoting the quantum state prepared by $P_N$. If the measurements of all parties are projective and they observe the output $\mathbf{a_1}$ given the input $\mathbf{x_1}$, then the post-measurement states are given by
\begin{eqnarray}
    \rho'_{\mathbf{a_1},\mathbf{x_1}}=\frac{\left(\bigotimes_{n=1}^N M^n_{a_{n,1}|x_{n,1}}\right)\rho_{N}\left(\bigotimes_{n=1}^N M^n_{a_{n,1}|x_{n,1}}\right)}{ p(\mathbf{a_1}|\mathbf{x_1},P_N)}.
\end{eqnarray}
Then, the particles interact via some unitary $V_N$ to evolve to the post-interaction states $\sigma_{\mathbf{a_1},\mathbf{x_1}}$. We consider that the unitary $V_N$ maps quantum states from the Hilbert space $\bigotimes_n\mathcal{H}_{A_n(t_1)}$ to some other Hilbert space $\bigotimes_n\mathcal{H}_{A_n(t_2)}$.

\textit{Model-independent inference for $N$ systems---} Let us consider a reference experiment giving rise to the same correlations $\vec{p}_{N,1}$ when some known observables $A'_{n,l}$ for $n=1,\ldots,N$ and $l=0,1$ are performed on a known quantum state prepared by the source 
$\ket{\psi'_{N}}\in\bigotimes_n\mathcal{H}_{A'_n(t_1)}$. Then the post-measured states interact via some known unitary $V'_N$ after which they are measured using the same known observables $A'_{n,l}$ to obtain $\vec{p}_{N,2,\mathbf{a_1},\mathbf{x_1}}$. If one observes the same correlations in the actual experiment, then it is equivalent to the reference one as: (i) the local Hilbert spaces admit the product form $\mathcal{H}_{A_n(t_i)}=\mathcal{H}_{A_n'(t_i)}\otimes\mathcal{H}_{A_n''(t_i)}$ for some auxiliary Hilbert spaces $\mathcal{H}_{A_n''(t_i)}$; (ii) there are local unitary operations  $U_{n(t_i)}:\mathcal{H}_{A_n(t_i)}\to \mathcal{H}_{A_n'(t_i)}\otimes\mathcal{H}_{A_n''(t_i)}$ for $i=1,2$
such that
\begin{equation}\label{ststateNm}
 \left(\bigotimes_nU_{n(t_1)}\right)\rho_{AB}\left(\bigotimes_nU_{n(t_1)}\right)^{\dagger}=\proj{\psi'_N}\otimes\xi_{\mathbf{A''(t_1)}}
\end{equation}
where  $\xi_{\mathbf{A_1''(t_1)}}$ acts on $\bigotimes_n\mathcal{H}_{A''_n(t_i)}$ is some auxiliary quantum state, and
\begin{eqnarray}\label{stmeaNm}
U_{n(t_i)}\,\overline{A}_{n,l(t_i)}\,U_{n(t_1)}&=&A_{n,l}'\otimes\mathbbm{1}_{A_{n}''(t_i)}
\end{eqnarray}
where $\mathbbm{1}_{A_n''(t_i)}$ is the identity acting on the parties auxiliary system and $\overline{A}_{{n,l}(t_i)}=\Pi^n_{i}A_{n,l}\Pi^n_{i}$ such that $\Pi^n_{i}$ denotes the projection of the observables $A_{n,l}$ onto the Hilbert space $\mathcal{H}_{A_n(t_i)}$; (iii) The interaction $V_N$ is certified as 
\begin{eqnarray}\label{stintNm}
   \left(\bigotimes_nU_{n(t_2)}\right)\ V_N\left(\bigotimes_nU_{n(t_1)}\right)^{\dagger}=V'_N\otimes V_{\mathrm{aux}}
\end{eqnarray}
where $V_{\mathrm{aux}}$ is a unitary matrix mapping $\bigotimes_n\mathcal{H}_{A_n(t_1)''}$ to $\bigotimes_n\mathcal{H}_{A_n(t_2)''}$.

Consider now  the following Bell inequalities for any $\mathbf{a_1}$ introduced in \cite{sarkar2023universal}
\begin{eqnarray}\label{BE1NNm}
\mathcal{B}_{\mathbf{a_1}}=(-1)^{a_{1,1}} \left\langle(N-1)\tilde{A}_{1,1}\otimes\bigotimes_{n=2}^N A_{n,1} \qquad\quad\right.\nonumber\\\left.+\sum_{i=2}^N(-1)^{a_{n,1}}\tilde{A}_{1,0}\otimes A_{n,0}\right\rangle\leqslant  \beta_C
\end{eqnarray}
where $A_{n,1}$ are quantum observables defined via the measurement operators of the $n-$th party as in \eqref{obspic} and
\begin{eqnarray}\label{overANm}
    \tilde{A}_{1,0}=\frac{A_{0}-A_{1}}{\sqrt{2}},\qquad\tilde{A}_{1,1}=\frac{A_{0}+A_{1}}{\sqrt{2}}.
\end{eqnarray}
The classical bound of the above Bell inequalities is $\beta_C=\sqrt{2}(N-1)$ for any $\mathbf{a_1}$. 
Consider now the following states 
\begin{equation}\label{GHZvecsNm}
\ket{\phi_{\mathbf{a_1}}}=\frac{1}{\sqrt{2}}(\ket{a_{1,1}\ldots a_{N,1}}+(-1)^{a_{1,1}}\ket{a_{1,1}^{\perp}\ldots a_{N,1}^{\perp}}),
\end{equation}
where $a_{n,1}^{\perp}=1-a_{n,1}$, and the following observables for $n=2,\ldots,N$
\begin{equation}\label{GHZObsNm}
A_{1,0}=\frac{X+Z}{\sqrt{2}},\quad A_{1,1}= \frac{X-Z}{\sqrt{2}}, \quad A_{n,0}=Z,\quad B_{n,1}=X. 
\end{equation}
As shown in Fact 2 of Appendix B, using these states \eqref{GHZvecsNm} and observables \eqref{GHZObsNm}, one can attain the quantum bound $\beta_Q=2(N-1)$ of all the Bell inequalities \eqref{BE1NNm}.

Let us now suppose that correlations $\vec{p}_{N,1}$ achieves the quantum bound of $\mathcal{B}_{0,\ldots,0}$. Furthermore, the correlations $\vec{p}_{N,2,\mathbf{a_1},0,0,1,\ldots,1}$ achieves the quantum bound of $\mathcal{B}_{\mathbf{a_1}}$ for every $\mathbf{a_1}$. Along with them, suppose one also observes that when $a_{n,1}=0$ for all $n$ and $x_{1,1}=x_{2,1}=1,\ x_{n,1}=0 (n=3,\ldots,N)$ that the post-interaction state satisfies
\begin{eqnarray}\label{extrastatNm}
    -\langle \tilde{A}_{1,0}\otimes\I\rangle=\langle\I\otimes A_{n,1}\rangle=1.\qquad n=2,\ldots,N.
\end{eqnarray}

Consider now the following unitary
\begin{eqnarray}\label{UNm}
\mathcal{U_N}=\sum_{\mathbf{a_1}=0,1}\ket{\phi_{\mathbf{a_1}}}\!\bra{\overline{a}_{1,1}a_{2,1}a_{3,1}^x\ldots a_{N,1}^x}
\end{eqnarray}
where $\{\ket{\overline{0}},\ket{\overline{1}}\}$ and $\{\ket{0^x},\ket{1^x}\}$ are the eigenvectors of $(X+Z)/\sqrt{2}$ and $X$ respectively. It is straightforward to verify that if the state $\ket{\phi_{0,\ldots,0}}$ \eqref{GHZvecsNm} after being measured by the observables \eqref{GHZObsNm} and obtaining outcome $\mathbf{a_1}$ with inputs $x_{n,1}=0$ evolves via the unitary \eqref{UNm}, then one obtains the post-interaction state as $\ket{\phi_{\mathbf{a_1}}}$. Consequently, all these states and observables satisfy the above-mentioned statistics and thus we take them as the reference quantum states, observables and interaction.

Let us now state the result.

\begin{thm}\label{theorem1Nm}
Assume that the Bell inequality $\mathcal{B}_{0,\ldots,0}$ \eqref{BE1NNm} is maximally violated at $t_1$. Furthermore, when the parties obtain the outcome $\mathbf{a_1}$ with inputs $x_{1,1}=x_{2,1}=0$ and $x_{n,1}=1(n=3,\ldots,N)$, then the Bell inequalities $\mathcal{B}_{\mathbf{a_1}}$ \eqref{BE1NNm} for any $\mathbf{a_1}$ are maximally violated at $t_2$. Along with it, when the parties observe the outcomes $a_{n,1}=0$ for all $n$ and $x_{1,1}=x_{2,1}=1,\ x_{n,1}=0 (n=3,\ldots,N)$, then the condition \eqref{extrastatNm} is also satisfied. Then, the quantum state prepared by the source and the observables of both parties are certified as in \eqref{ststateNm} and \eqref{stmeaNm}, with the reference strategies given below Eq. \eqref{UNm}. Importantly, the unitary $V_N$ is certified as defined in \eqref{stintNm}
\begin{eqnarray}\label{VNm}
      \left(\bigotimes_nU_{n(t_2)}\right)\ V_N\left(\bigotimes_nU_{n(t_1)}\right)^{\dagger}=\mathcal{U_N}\otimes V_{\mathrm{aux}}
\end{eqnarray}
where $\mathcal{U_N}$ is given in Eq. \eqref{UNm} and $V_{\mathrm{aux}}$ is unitary.
\end{thm}

The proof of the above theorem follows similar lines as the two-system interaction case and is presented in Appendix B. 

\textit{Discussions---} The controlled-not gate which is an entangling quantum operation has been self-tested in \cite{pavel}. However, their scheme required not applying the gate at half the rounds of the experiment, thus making it not applicable to certify the interaction between two systems as one would be required to assume a particular model of the interaction to impose it, or put it simply, one can not switch off the interaction between two systems without assuming the nature of the interaction. Then, the scheme of \cite{pavel} also needed two independent sources which is again an additional assumption that can not be verified. Here, we do not make any such assumptions thus making our scheme model-independent. Although our proposed scheme currently has limited applicability, it can be considered as an alternative to infer quantum interaction rather than the standard approach to observe scattering amplitudes. 

For a remark, one might argue that the measurement device might cheat by not sending the post-measurement states at $t_1$ but sending some arbitrary ones. However, this can be easily verified by randomly choosing some rounds of the experiment where instead of sending the post-measured systems to interact, one measures it again in the same lab. As the measurements are certified to be projective, if one obtains any other outcome than the one observed at $t_1$ then the device is cheating. Furthermore, if one assumes that the interaction does not change the local supports of the post-measured states, then the source, instead of preparing a maximally entangled state, can even send local white noise to both the detectors and the scheme would still allow certification of the interaction between the post-measured states. 

Several follow-up problems arise from our work. The first problem would be to compute the robustness of the proposed scheme towards experimental errors. Then, one could explore the other classes of interaction that could be certified in a model-independent way. Furthermore, one might extend this approach where one could also certify time-dependent evolutions.

\textit{Acknowledgements---}
 This project was funded within the QuantERA II Programme (VERIqTAS project) that has received funding from the European Union’s Horizon 2020 research and innovation programme under Grant Agreement No 101017733 and from the Polish National Science Center (project No 2021/03/Y/ST2/00175).

\providecommand{\noopsort}[1]{}\providecommand{\singleletter}[1]{#1}%

\appendix

\onecolumngrid

\section{Appendix A: Model-independent inference of two-system entangling quantum interaction}

\subsection{Reference statistics}

Consider the following reference interaction that acts between the incoming particles given by the unitary
\begin{eqnarray}\label{U2}
\mathcal{U}=\ket{\phi_{0,0}}\!\bra{\overline{0}0}+\ket{\psi_{1,0}}\!\bra{\overline{1}0}+\ket{\psi_{0,1}}\!\bra{\overline{0}1}+\ket{\phi_{1,1}}\!\bra{\overline{1}1}
\end{eqnarray}
where, $\ket{\phi_{a_1,b_1}}$ are given in \eqref{GHZvecsm} and $\{\ket{\overline{0}},\ket{\overline{1}}\}$ are the eigenvectors of $(X+Z)/\sqrt{2}$ given in the computational basis as
\begin{eqnarray}\label{X+Z}
\ket{\overline{0}}=\cos(\pi/8)\ket{0}+\sin(\pi/8)\ket{1},\qquad \ket{\overline{1}}=-\sin(\pi/8)\ket{0}+\cos(\pi/8)\ket{1}.
\end{eqnarray}

It is clear from the above unitary that for some systems initially in product states will become entangled if they interact via this unitary process.


Consider now the scenario depicted in Fig. \ref{fig1} and that the source prepares some state $\rho_{AB}$ on which Alice and Bob measure at time $t_1$. Let us suppose that Alice and Bob choose the input $x_1,y_1$ and obtain the outcomes $a_1,b_1$ respectively. The post-measurement states are denoted as $\rho'_{a_1,b_1,x_1,y_1}$. Then, the states at time $t_1$ evolve via some unitary process $V$ \cite{Sakurai}.
Now, the post-interaction states $\sigma_{a_1,b_1,x_1,y_1}$ for $a_1,b_1=0,1$ and $x_1=y_1=0$ need to maximally violate the following Bell inequalities
\begin{eqnarray}\label{BE1N}
\mathcal{B}_{a_1,b_1}=(-1)^{a_1} \left\langle\tilde{A}_{1}\otimes B_{1} +(-1)^{b_1}\tilde{A}_{0}\otimes B_{0} \right\rangle\leqslant  \beta_C,
\end{eqnarray}
where,
\begin{eqnarray}\label{overA}
    \tilde{A}_{0}=\frac{A_{0}-A_{1}}{\sqrt{2}},\qquad\tilde{A}_{1}=\frac{A_{0}+A_{1}}{\sqrt{2}}.
\end{eqnarray}

\begin{fakt}The maximal quantum value of the Bell expression $\mathcal{B}_{a_1,b_1}$ is $2$ and it is achieved by the following observables 
\begin{eqnarray}\label{GHZObs}
A_{0}=\frac{X+Z}{\sqrt{2}},\quad A_{1}= \frac{X-Z}{\sqrt{2}}, \quad B_{0}=Z,\quad B_{1}=X 
\end{eqnarray}
as well as the maximally entangled states 
\begin{equation}\label{GHZvecs}
\ket{\phi_{a_1,b_1}}=\frac{1}{\sqrt{2}}(\ket{a_1 b_1}+(-1)^{a_{1}}|a_1^{\perp}b_1^{\perp}\rangle),
\end{equation}
where $a_1^{\perp}=1-a_1,b_1^{\perp}=1-b_1$. Notice that $\ket{\phi_{0,0}}=\ket{\phi^+}, \ket{\phi_{1,0}}=-\ket{\psi^-},\ket{\phi_{0,1}}=\ket{\psi^+},\ket{\phi_{1,1}}=-\ket{\phi^-}$.
\end{fakt}

\begin{proof}
The following proof is inspired by Ref. \cite{sarkar2023universal}. The Bell operators $\hat{\mathcal{B}}_{a_1,b_1}$ corresponding to each of the Bell expressions $\mathcal{B}_{a_1,b_1}$ \eqref{BE1N} admit a sum-of-squares (SOS) decomposition of the following form
\begin{eqnarray}\label{SOS1}
   2\left[2\I-\mathcal{\hat{B}}_{a_1,b_1}\right]\geq P_{a_1}^2+Q_{a_1,b_1}^2, 
\end{eqnarray}
where 

\begin{equation}\label{SOS2}
    P_{a_1}=(-1)^{a_1}\tilde{A}_{1}- B_{1},\qquad
   Q_{a_1,b_1}=(-1)^{a_1+b_1}\tilde{A}_{0}-B_{0}.
\end{equation}
Let us verify the above decomposition \eqref{SOS1} by expanding the terms on the right-hand side Eq.  \eqref{SOS1} which gives us
\begin{eqnarray}\label{SOS11}
 \left(\tilde{A}_{1}^2+\tilde{A}_{0}^2+B_{1}^2+B_0^2\right)-2\ \mathcal{\hat{B}}_{l}\leqslant \  2\left[2\ \I-\ \mathcal{\hat{B}}_{l}\right],
\end{eqnarray}
where we used the fact that $\tilde{A}_{i}^2\leqslant \I$. Consequently, we have from Eq. \eqref{SOS11} that $2$ is an upper bound on $\mathcal{B}_{a_1,b_1}$ in quantum theory. Now, it is easy to verify that the Bell inequalities \eqref{BE1N} attain the value $2$ when the states shared among the parties are the maximally entangled states $\ket{\phi_{a_1,b_1}}$ \eqref{GHZvecs} and the observables given in Eq. (\ref{GHZObs}). Thus, the Tsirelson's bound of the Bell inequalities \eqref{BE1N} for any $a_1,b_1$ is $2$.
\end{proof}

An important observation from the above SOS decomposition is that when the maximal violation of the above Bell inequalities is attained, the right-hand side of Eq. \eqref{SOS1} is $0$. Thus, any state $\ket{\psi}$ and observables $A_{i},B_i$ that attain the maximal violation of the Bell inequalities \eqref{BE1N} for any $l$ is given by
    \begin{equation}\label{SOSrel1}
          (-1)^{a_1}\tilde{A}_{1}\ket{\psi}= B_{1}\ket{\psi},\qquad
   (-1)^{a_1+b_1}\tilde{A}_{0}\ket{\psi}=B_{0}\ket{\psi},
    \end{equation}
Furthermore, one can also obtain from Eq. \eqref{SOS11} that the observables that attain the maximum violation of the Bell inequalities \eqref{BE1N} must be unitary, that is, $A_{i}^2=B_i^2=\I$. These relations will be particularly useful for self-testing.

Along with the above statistics, one also needs to observe when $a_1=b_1=0$ and $x_1=1,y_1=1$ that the post-interaction state satisfies
\begin{eqnarray}\label{extrastat}
    \langle\I\otimes B_{1}\rangle= -\langle \tilde{A}_{0}\otimes\I\rangle=1.
\end{eqnarray}

\subsection{Model-independent certification}
\setcounter{thm}{0}
\begin{thm}\label{theorem1}
Assume that the Bell inequality $\mathcal{B}_{0,0}$ \eqref{BE1N} is maximally violated at $t_1$. Furthermore, when Alice and Bob obtain the outcome $a_1,b_1$ with inputs $x_1=y_1=0$, then the Bell inequalities $\mathcal{B}_{a_1,b_1}$ \eqref{BE1N} for any $a_1,b_1$ are maximally violated at $t_2$. Along with it, when Alice and Bob observe the outcomes $a_1=b_1=0$ with inputs $x_1=1,y_1=1$, then the condition \eqref{extrastat} is also satisfied. The source $P$ prepares the state $\rho_{AB}$ acting on $\mathcal{H}_{A(t_1)}\otimes\mathcal{H}_{B(t_1)}$ and the interaction between the particles is represented by the unitary $V:\mathcal{H}_{A(t_1)}\otimes\mathcal{H}_{B(t_1)}\rightarrow\mathcal{H}_{A(t_2)}\otimes\mathcal{H}_{B(t_2)}$. Furthermore, the local observables of Alice and Bob are given by $A_{i},B_i\ (i=0,1)$ that acts on $\mathcal{H}_{A},\mathcal{H}_{B}$ respectively where $\mathcal{H}_{A}\equiv\mathcal{H}_{A(t_1)}\cup\mathcal{H}_{A(t_2)},\mathcal{H}_{B}\equiv\mathcal{H}_{B(t_1)}\cup\mathcal{H}_{B(t_2)}$. Then, there exist local unitary transformations $U_{A(t_i)}:\mathcal{H}_{A(t_i)}\rightarrow(\mathbb{C}^2)_{A'(t_i)}\otimes\mathcal{H}_{A''(t_i)}$ and $U_{B(t_i)}:\mathcal{H}_{B(t_i)}\rightarrow(\mathbb{C}^2)_{B'(t_i)}\otimes \mathcal{H}_{B''(t_i)}$ for $i=1,2$ 
such that:
\begin{enumerate}
    \item  The state $\rho_{AB}$ is certified as
\begin{eqnarray}\label{A10}
U_{A(t_1)}\otimes U_{B(t_1)}\rho_{AB}U_{A(t_1)}^{\dagger}\otimes U_{B(t_1)}^{\dagger}=\proj{\phi^+}_{A'(t_1)B'(t_1)}\otimes\xi_{A''(t_1)B''(t_1)}.
\end{eqnarray}
\item The observables of Alice and Bob are certified as
\begin{eqnarray}\label{mea1}
U_{A(t_i)}\overline{A}_{0(t_i)}\,U_{A(t_i)}^{\dagger}&=&\left(\frac{X+Z}{\sqrt{2}}\right)_{A'(t_i)}\otimes\I_{A''(t_i)}, \qquad U_{A(t_i)}\overline{A}_{1(t_i)}\,U_{A(t_i)}^{\dagger}=\left(\frac{X-Z}{\sqrt{2}}\right)_{A'(t_i)}\otimes\I_{A''(t_i)},\nonumber\\
U_{B(t_i)}\overline{B}_{0(t_i)}\,U_{B(t_i)}^{\dagger}&=&Z_{B'(t_i)}\otimes\I_{B''(t_i)},\qquad U_{B(t_i)}\overline{B}_{1(t_i)}\,U_{B(t_i)}^{\dagger}=X_{B'(t_i)}\otimes\I_{B''(t_i)}
\end{eqnarray}
where $\overline{A}_{j(t_i)}=\Pi^A_{i}A_{j}\Pi^A_{i}$ and  $\overline{B}_{j(t_i)}=\Pi^B_{i}B_{j}\Pi^B_{i}$ such that $\Pi^A_{i}/\Pi^B_{i}$ denotes the projection onto the Hilbert space $\mathcal{H}_{A(t_i)}/\mathcal{H}_{B(t_i)}$.

\item The unitary $V$ is certified as $U_{A(t_2)}\otimes U_{B(t_2)}\ V\ U_{A(t_1)}^{\dagger}\otimes U_{B(t_1)}^{\dagger}=\mathcal{U}\otimes V_0$ where $\mathcal{U}$ is given in Eq. \eqref{U2} and $V_0$ is a unitary matrix mapping $\mathcal{H}_{A''(t_1)}\otimes\mathcal{H}_{B''(t_1)}$ to $\mathcal{H}_{A''(t_2)}\otimes\mathcal{H}_{B''(t_2)}$.
\end{enumerate}
\end{thm}
\begin{proof}
The proof is divided into two major parts. In the first part, we certify the local observables of Alice and Bob along with the states prepared by the source $P$ and the post-interaction states. Then in the second part, using the certified states and observables we certify the unitary $V$.
    \begin{center}
        \textbf{Local observables and the state prepared by the source $P$}
    \end{center}
Consider the scenario presented in Fig. 1 of the manuscript at time $t_1$, that is, before the interaction. Now, the quantum state $\rho_{AB}$ acts on $\mathcal{H}_{A(t_1)}\otimes\mathcal{H}_{B(t_1)}$ and thus we consider the projection of the observables $A_{j}, B_{j}$ onto the support of the state as $\overline{A}_{j(t_1)}, \overline{B}_{j(t_1)}$ where
\begin{eqnarray}
    \overline{A}_{j(t_1)}=\Pi^A_{1}A_{j}\Pi^A_{1},\qquad \overline{B}_{j(t_1)}=\Pi^B_{1}B_{j}\Pi^B_{1}\qquad (j=0,1)
\end{eqnarray}
such that $\Pi^A_{1},\Pi^B_{1}$ denote the projectors onto $\mathcal{H}_{A(t_1)},\mathcal{H}_{B(t_1)}$ respectively.

Following \cite{sarkar2023universal} for $N=2$ and $a_1=b_1=0$, we obtain from Eqs. \eqref{SOSrel1} that
\begin{eqnarray}
    \left\{\overline{{A}}_{0,1},\overline{{A}}_{1,1}\right\}=0,\qquad \left\{\overline{{B}}_{0,1},\overline{{B}}_{1,1}\right\}=0.
\end{eqnarray}
Consequently, there exists a unitary transformation $U_{s(t_1)}:\mathcal{H}_{s(t_1)}\rightarrow(\mathbb{C}^2)_{s'(t_1)}\otimes\mathcal{H}_{s''(t_1)}$ $(s=A,B)$ to obtain \eqref{mea1} for $i=1$. Utilising now the relation \eqref{SOSrel1} for $a_1=b_1=0$ and following exactly from \cite{sarkar2023universal}, we obtain that
\begin{eqnarray}
   U_{A(t_1)}\otimes U_{B(t_1)}\rho_{AB}U_{A(t_1)}^{\dagger}\otimes U_{B(t_1)}^{\dagger}=\proj{\phi^+}_{A'(t_1)B'(t_1)}\otimes\xi_{A''(t_1)B''(t_1)}.
\end{eqnarray}

Now as depicted in Fig. 1 of the manuscript, the post-measurement states, $\rho'_{a_1,b_1,x_1,y_1}$, of Alice and Bob undergo a joint unitary transformation $V:\mathcal{H}_{A(t_1)}\otimes\mathcal{H}_{B(t_1)}\rightarrow\mathcal{H}_{A(t_2)}\otimes\mathcal{H}_{B(t_2)}$ to get the post-interaction states, $\sigma_{a_1,b_1,x_1,y_1}$.
As the measurements and states at time $t_1$ are self-tested, 
the post-measurement states, $\rho'_{a_1,b_1,x_1,y_1}$ are given by
\begin{eqnarray}\label{pre-int state}
U_{A(t_1)}\otimes U_{B(t_1)}\rho'_{a_1,b_1,x_1,y_1}U_{A(t_1)}^{\dagger}\otimes U_{B(t_1)}^{\dagger}=\proj{e_{a_1,x_1}}_{A'(t_1)}\otimes\proj{e_{b_1,y_1}}_{B'(t_1)}\otimes\xi_{A''(t_1)B''(t_1)}
\end{eqnarray}
where $\ket{e_{a_1/b_1,x_1/y_1}}$ is the eigenvector corresponding to the $a_1/b_1$ outcome of $x_1/y_1$ reference measurement.

Consider now the scenario presented in Fig. 1 of the manuscript at time $t_2$, that is, after the interaction.
The post-interaction states $\sigma_{a_1,b_1,0,0}$ maximally violate the Bell inequalities \eqref{BE1N} for any $a_1,b_1$. Consequently, they are also certified from the maximal violation of the Bell inequalities \eqref{BE1N}. For this purpose, we again utilise the relation \eqref{SOSrel1} for all $a_1,b_1$ and following exactly from \cite{sarkar2023universal}, we first obtain that
\begin{eqnarray}
    \left\{\mathbb{A}_{0}^{(a_1,b_1)},\mathbb{A}_{1}^{(a_1,b_1)}\right\}=0,\qquad \left\{\mathbb{B}_{0}^{(a_1,b_1)},\mathbb{B}_{1}^{(a_1,b_1)}\right\}=0\qquad \forall a_1,b_1
\end{eqnarray}
where $\mathbb{A}_{j}^{(a_1,b_1)}=\Pi_{A}^{(a_1,b_1)}A_j\Pi_{A}^{(a_1,b_1)}$ and 
$\mathbb{B}_{j}^{(a_1,b_1)}=\Pi_{B}^{(a_1,b_1)}B_j\Pi_{B}^{(a_1,b_1)}$ for $j=0,1$ such that $\Pi_{s}^{(a_1,b_1)}$ is the projection onto the local supports of the states $\rho_{s}=\Tr_{A,B/s}(\sigma_{a_1,b_1,0,0})$ for $s=A,B$. Now, utilizing Lemma 2 of \cite{sarkar2023self}, we obtain that the local observables acting on the support of all states, the corresponding projector is given by $\Pi_2^s=\cup_{a_1,b_1}\Pi_{s}^{(a_1,b_1)}$, also anti-commute 
\begin{eqnarray}
    \{\overline{A}_{0(t_2)},\overline{A}_{1(t_2)}\}=0,\qquad\{\overline{B}_{0(t_2)},\overline{B}_{0(t_2)}\}=0
\end{eqnarray}
such that $\overline{A}_{j(t_2)}=\Pi_2^AA_{j}\Pi_2^A$ and $\overline{B}_{j(t_2)}=\Pi_2^BB_{j}\Pi_2^B$. Consequently, one can straightaway conclude that there exists unitaries $U_{s(t_2)}:\mathcal{H}_{s(t_2)}\rightarrow(\mathbb{C}^2)_{s'(t_2)}\otimes\mathcal{H}_{s''(t_2)}$ $(s=A,B)$ to obtain \eqref{mea1} for $i=2$. Furthermore from \cite{sarkar2023universal}, the post-interaction states are certified using the relations \eqref{SOSrel1} for all $a_1,b_1$ using the certified observables \eqref{mea1} at time $t_2$ as
%

\begin{eqnarray}\label{post-int state}
   U_{A(t_2)}\otimes U_{B(t_2)}\sigma_{a_1,b_1,0,0}U_{A(t_2)}^{\dagger}\otimes U_{B(t_2)}^{\dagger}=\proj{\phi_{a_1,b_1}}_{A'(t_2)B'(t_2)}\otimes\xi^{a_1,b_1}_{A''(t_2)B''(t_2)}
\end{eqnarray}
where $\ket{\phi_{a_1,b_1}}$ are given in \eqref{GHZvecs}. Let us now utilize the certified pre-interaction states \eqref{pre-int state} and post-interaction ones \eqref{post-int state} to self-test the unitary $V$.
   \begin{center}
        \textbf{The unitary $V$}
    \end{center}
As shown above, the Hilbert spaces $\mathcal{H}_{s(t_i)}$ decompose as $(\mathbb{C}^2)_{s'(t_i)}\otimes\mathcal{H}_{s''(t_i)}$ for $s=A,B$ and $i=0,1$. Thus, we can express $V$ as
\begin{eqnarray}\label{Vgen}
   U_{A(t_2)}\otimes U_{B(t_2)}\ V\ U_{A(t_1)}^{\dagger}\otimes U_{B(t_1)}^{\dagger}=\sum_{i,i',j,j'}\ket{\overline{i}_{A'(t_1)}\ j_{B'(t_1)}}\bra{\overline{i}'_{A'(t_2)}\ j'_{B'(t_2)}}\otimes (V_{i,i'j,j'})_{A''(t_1)B''(t_1)A''(t_2)A''(t_2)}
\end{eqnarray}
where $\ket{\overline{i}}$ are the eigenvectors of $(X+Z)/\sqrt{2}$ as given in \eqref{X+Z}. For simplicity, we will drop the lower indices from Eq. \eqref{Vgen}.

Now, suppose when $x_1,y_1=0$ and $a_1=b_1=0$, then as described above the Bell inequality $\mathcal{B}_{0,0}$ is maximally violated. Thus from the relation $V\rho'_{a_1,b_1,x_1,y_1}V^{\dagger}=\sigma_{a_1,b_1,x_1,y_1}$, we can purify the states on the left and right sides by adding ancillary systems $E$ to obtain that 
\begin{eqnarray}
V\otimes\I_E\ket{\rho'_{0,0,0,0}}_{A(t_1)B(t_1)E}=\ket{\sigma_{0,0,0,0}}_{A(t_2)B(t_2)E}
\end{eqnarray}
such that $\Tr_E (\proj{\rho'_{0,0,0,0}})=\rho'_{0,0,0,0}$ and $\Tr_E (\proj{\sigma_{0,0,0,0}})=\sigma_{0,0,0,0}$.
Now, using the certified states $\rho'_{0,0,0,0}$ and $\sigma_{0,0,0,0}$ from Eqs. \eqref{pre-int state} and \eqref{post-int state} respectively and then expanding $V$ using \eqref{Vgen}, we obtain
\begin{eqnarray}
    \sum_{i,j}\ket{\overline{i}\ j} \left(V_{i,0,j,0}\otimes\I_E\ket{\xi}\right)=\ket{\phi_{00}}\otimes\ket{\xi^{0,0}}.
\end{eqnarray}
Multiplying by $\bra{\overline{i}\ j}$ for all $i,j$ on both sides of the above formula gives us the following four conditions
\begin{eqnarray}
\sqrt{2}\ V_{0,0,0,0}\otimes\I_E\ket{\xi}&=&\cos(\pi/8)\ket{\xi^{0,0}},\qquad \sqrt{2}\ V_{1,0,0,0}\otimes\I_E\ket{\xi}=-\sin(\pi/8)\ket{\xi^{0,0}}\nonumber\\
\sqrt{2}\ V_{0,0,1,0}\otimes\I_E\ket{\xi}&=&\sin(\pi/8)\ket{\xi^{0,0}},\qquad \sqrt{2}\ V_{1,0,1,0}\otimes\I_E\ket{\xi}=\cos(\pi/8)\ket{\xi^{0,0}}.
\end{eqnarray}
Thus, from the above conditions we can conclude that
\begin{eqnarray}
    \frac{1}{\cos(\pi/8)}V_{0,0,0,0}\otimes\I_E\ket{\xi}=-\frac{1}{\sin(\pi/8)}V_{1,0,0,0}\otimes\I_E\ket{\xi}= \frac{1}{\sin(\pi/8)}V_{0,0,1,0}\otimes\I_E\ket{\xi}\nonumber\\=\frac{1}{\cos(\pi/8)}V_{1,0,1,0}\otimes\I_E\ket{\xi}.
\end{eqnarray}
Now, without loss of generality, one can consider that $\xi$ is full-rank as the unitary $V$ can be characterised only on the support of the states $\rho'$. Thus, from the above formula taking partial trace over $E$ gives us
\begin{eqnarray}
     \frac{\sqrt{2}\ V_{0,0,0,0}}{\cos(\pi/8)}=\frac{-\sqrt{2}\ V_{1,0,0,0}}{\sin(\pi/8)}= \frac{\sqrt{2}\ V_{0,0,1,0}}{\sin(\pi/8)}=\frac{\sqrt{2}\ V_{1,0,1,0}}{\cos(\pi/8)}=V_{0,0}.
\end{eqnarray}

Similarly, considering the other outputs $a_1,b_1$ and $x_1=y_1=0$ and recalling that 
\begin{eqnarray}\label{36}
V\otimes\I_E\ket{\rho'_{a_1,b_1,0,0}}_{A(t_1)B(t_1)E}=\ket{\sigma_{a_1,b_1,0,0}}_{A(t_2)B(t_2)E}
\end{eqnarray}
where the states $\ket{\rho'_{a_1,b_1,0,0}}$ and $\ket{\sigma_{a_1,b_1,0,0}}$ are certified in Eqs. \eqref{pre-int state} and \eqref{post-int state} respectively. Consequently, from the general expression of the unitary $V$ \eqref{Vgen} and the above condition \eqref{36}, we obtain that
\begin{eqnarray}\label{Vab}
     \sum_{i,j}\ket{\overline{i}\ j} \left(V_{i,a_1,j,b_1}\otimes\I_E\ket{\xi}\right)=\ket{\phi_{a_1,b_1}}\otimes\ket{\xi^{a_1,b_1}}.
\end{eqnarray}
Multiplying by $\bra{\overline{i}\ j}$ for all $i,j$ on both sides of the above formula gives us the following conditions for any $a_1,b_1$
\begin{eqnarray}
V_{i,a_1,j,b_1}\otimes\I_E\ket{\xi}&=&\bra{\overline{i}\ j}\phi_{a_1,b_1}\rangle\ket{\xi^{a_1,b_1}},\qquad i,j=0,1
\end{eqnarray}
Thus, from the above conditions we can conclude that
\begin{eqnarray}
    \frac{1}{\bra{\overline{0}0}\phi_{a_1,b_1}\rangle}V_{0,a_1,0,b_1}\otimes\I_E\ket{\xi}= \frac{1}{\bra{\overline{0}1}\phi_{a_1,b_1}\rangle}V_{0,a_1,1,b_1}\otimes\I_E\ket{\xi}=  \frac{1}{\bra{\overline{1}0}\phi_{a_1,b_1}\rangle}V_{1,a_1,0,b_1}\otimes\I_E\ket{\xi}\nonumber\\= \frac{1}{\bra{\overline{1}1}\phi_{a_1,b_1}\rangle}V_{1,a_1,1,b_1}\otimes\I_E\ket{\xi}.
\end{eqnarray}
Again considering that $\xi$ is full-rank, from the above formula taking partial trace over $E$ gives us
\begin{eqnarray}
      \frac{V_{0,a_1,0,b_1}}{\bra{\overline{0}0}\phi_{a_1,b_1}\rangle}= \frac{V_{0,a_1,1,b_1}}{\bra{\overline{0}1}\phi_{a_1,b_1}\rangle}=  \frac{V_{1,a_1,0,b_1}}{\bra{\overline{1}0}\phi_{a_1,b_1}\rangle}= \frac{V_{1,a_1,1,b_1}}{\bra{\overline{1}1}\phi_{a_1,b_1}\rangle}=V_{a_1,b_1}.
\end{eqnarray}
Now, using the above condition in Eq. \eqref{Vgen} gives us 
\begin{eqnarray}\label{Vgen1}
    U_{A(t_2)}\otimes U_{B(t_2)}\ V\ U_{A(t_1)}^{\dagger}\otimes U_{B(t_1)}^{\dagger}&=&\sum_{i,a_1,j,b_1}\ket{\overline{i}\ j}\bra{\overline{a}_1\ b_1}\bra{\overline{i}\ j}\phi_{a_1,b_1}\rangle\otimes (V_{a_1,b_1})\nonumber\\
    &=&\sum_{i,a_1,j,b_1}\ket{\overline{i}\ j}\bra{\overline{i}\ j}\phi_{a_1,b_1}\rangle\bra{\overline{a}_1\ b_1}\otimes (V_{a_1,b_1})\nonumber\\&=&
    \ket{\phi_{0,0}}\!\bra{\overline{0}0}\otimes V_{00}+\ket{\phi_{1,0}}\!\bra{\overline{1}0}\otimes V_{10}+\ket{\phi_{0,1}}\!\bra{\overline{0}1}\otimes V_{01}+\ket{\phi_{1,1}}\!\bra{\overline{1}1}\otimes V_{11}
\end{eqnarray}
where we used the fact that$\sum_{ij}\ket{\overline{i}\ j}\bra{\overline{i}\ j}=\I$.


Now, we finally utilize the fact that when $a_1=b_1=0$ and $x_1=1,y_1=1$, then one satisfies the relation \eqref{extrastat}. Let us exploit this relation to show that $V_{00}=V_{10}=V_{01}=V_{11}$. The state $\sigma_{0,0,1,1}$ is given by
\begin{eqnarray}
U_{A(t_2)}\otimes U_{B(t_2)}  \ket{\sigma_{0,0,1,1}}= \left(U_{A(t_2)}\otimes U_{B(t_2)} V\otimes\I_E U_{A(t_1)}^{\dagger}\otimes U_{B(t_1)}^{\dagger}\right) U_{A(t_1)}\otimes U_{B(t_1)}\ket{\rho'_{0,0,1,1}}.
\end{eqnarray}
Using the fact that $\rho'_{0,0,1,1}$ is certified as in \eqref{pre-int state} and the form of $V$ from \eqref{Vgen1}, we obtain from the above formula that
\begin{eqnarray}
    U_{A(t_2)}\otimes U_{B(t_2)}  \ket{\sigma_{0,0,1,1}}=\frac{1}{2}\left(\ket{\phi_{0,0}}\otimes V_{00}+\ket{\phi_{1,0}}\otimes V_{10}+\ket{\phi_{0,1}}\otimes V_{01}+\ket{\phi_{1,1}}\otimes V_{11}\right)\otimes\I_E\ket{\xi}.
\end{eqnarray}
Consider now the relation $\langle \I\otimes B_1\rangle=1$, which using Cauchy-Schwarz inequality can be written as 
\begin{eqnarray}
    \I\otimes B_1\ket{\sigma_{0,0,1,1}}_{A(t_2)B(t_2)E}=\ket{\sigma_{0,0,1,1}}_{A(t_2)B(t_2)E}.
\end{eqnarray}
Notice that $\sigma_{0,0,1,1}$ acts on $\mathcal{H}_{A(t_2)}\otimes\mathcal{H}_{B(t_2)}$ and $B_1$ is certified on it as in Eq. \eqref{mea1} for $i=2$. Consequently, we  obtain from the above condition that
\begin{eqnarray}
  \big(\ket{\phi_{0,1}}\otimes V_{00}+\ket{\phi_{1,1}}\otimes V_{10}+\ket{\phi_{0,0}}\otimes V_{01}&+&\ket{\phi_{1,0}}\otimes V_{11}\big)\otimes\I_E\ket{\xi}\nonumber\\
  &=&\big(\ket{\phi_{0,0}}\otimes V_{00}+\ket{\phi_{1,0}}\otimes V_{10}+\ket{\phi_{0,1}}\otimes V_{01}+\ket{\phi_{1,1}}\otimes V_{11}\big)\otimes\I_E\ket{\xi}
\end{eqnarray}
which allows us to conclude that
\begin{eqnarray}
    V_{00}\otimes\I_E\ket{\xi}=V_{01}\otimes\I_E\ket{\xi}\qquad V_{10}\otimes\I_E\ket{\xi}=V_{11}\otimes\I_E\ket{\xi}.
\end{eqnarray}
Again considering $\xi$ to be full-rank, we obtain
\begin{eqnarray}\label{Vgen2}
    V_{00}=V_{01}\qquad V_{10}=V_{11}.
\end{eqnarray}
Now, utilising the relation $\langle \tilde{A_0}\otimes \I\rangle=-1$, which again using Cauchy-Schwarz inequality can be written as 
\begin{eqnarray}
\tilde{A_0}\otimes\I_{BE}\ket{\sigma_{0,0,1,1}}=-\ket{\sigma_{0,0,1,1}}
\end{eqnarray}
and then using the fact that $A_1,A_2$ is certified when acting on the Hilbert space $\mathcal{H}_{A(t_2)}\otimes\mathcal{H}_{B(t_2)}$ as in Eq. \eqref{stmea} which gives us
\begin{eqnarray}
  \big(\ket{\phi_{1,1}}\otimes V_{00}+\ket{\phi_{0,1}}\otimes V_{10}+\ket{\phi_{1,0}}\otimes V_{01}&+&\ket{\phi_{0,0}}\otimes V_{11}\big)\otimes\I_E\ket{\xi}\nonumber\\&=&\big(\ket{\phi_{0,0}}\otimes V_{00}+\ket{\phi_{1,0}}\otimes V_{10}+\ket{\phi_{0,1}}\otimes V_{01}+\ket{\phi_{1,1}}\otimes V_{11}\big)\otimes\I_E\ket{\xi}.
\end{eqnarray}
This allows to conclude that
\begin{eqnarray}\label{Vgen3}
    V_{00}=V_{11}\qquad V_{10}=V_{01}.
\end{eqnarray}
It is simple to observe from Eqs. \eqref{Vgen2} and \eqref{Vgen3} that $V_{00}=V_{10}=V_{01}=V_{11}=V_0$. Consequently, one can straightforwardly conclude from \eqref{Vgen1} that
\begin{eqnarray}
 U_{A(t_2)}\otimes U_{B(t_2)}\ V\ U_{A(t_1)}^{\dagger}\otimes U_{B(t_1)}^{\dagger}=\left(\ket{\phi_{0,0}}\!\bra{\overline{0}0}+\ket{\phi_{1,0}}\!\bra{\overline{1}0}+\ket{\phi_{0,1}}\!\bra{\overline{0}1}+\ket{\phi_{1,1}}\!\bra{\overline{1}1}\right)\otimes V_{0}.
\end{eqnarray}
This completes the proof.
\end{proof}

\section{Appendix B: Model-independent inference of $N$-system entangling quantum interaction}

\subsection{Reference statistics}
Consider the following reference interaction that acts between the $N$ incoming particles given by the unitary
\begin{eqnarray}\label{Nuni}
\mathcal{U_N}=\sum_{i_1,i_2,\ldots i_N=0,1}\ket{\phi_{i_1i_2\ldots i_N}}\!\bra{\overline{i}_1i_2i_{2}^x\ldots i_{N}^x}
\end{eqnarray}
where, the $\{\ket{\overline{0}},\ket{\overline{1}}\}$ are the eigenvectors of $(X+Z)/\sqrt{2}$ given Eq. \eqref{X+Z} and $\ket{0^x}=\ket{+}=(\ket{0}+\ket{1})/\sqrt{2},\ket{1^x}=\ket{-}=(\ket{0}-\ket{1})/\sqrt{2}$ with 
\begin{equation}
\ket{\phi_{l_1l_2\ldots l_N}}=\frac{1}{\sqrt{2}}(\ket{l_1\ldots l_N}+(-1)^{l_{1}}|l^{\perp}_1\ldots l^{\perp}_N\rangle),
\end{equation}
such that $l_1,l_2,\ldots,l_N=0,1$ and $l^{\perp}_i=1-l_i$. 
It is clear from the above unitary that for some systems initially in product states will become entangled if they interact via this unitary process. 

Consider now that the source prepares some state $\rho_{N}$ on which all parties measure at time $t_1$. Let us suppose that the parties choose the input $\mathbf{x_1}$ and obtain the outcomes $\mathbf{a_1}$ respectively. Now, the post-measurement states are denoted as $\rho'_{\mathbf{a_1},\mathbf{x_1}}$. Then, they interact via some Hamiltonian $H(t)$ for a time $\delta t$. Consequently, the states at time $t_1$ evolve via the unitary process $V$ \cite{Sakurai} and thus the states after the interaction is given by $V\rho'_{\mathbf{a_1},\mathbf{x_1}}V^{\dagger}=\sigma_{\mathbf{a_1},\mathbf{x_1}}$. 

Now, the states $\sigma_{\mathbf{a_1},\mathbf{x_1}}$ for any $\mathbf{a_1}$ and $x_{1,1}=x_{2,1}=0$ with $x_{n,1}=1\ (n=3,\ldots,N)$ need to maximally violate the following Bell inequalities
\begin{eqnarray}\label{BE1NN}
\mathcal{B}_{\mathbf{a_1}}=(-1)^{a_{1,1}} \left\langle(N-1)\tilde{A}_{1,1}\otimes\bigotimes_{n=2}^N A_{n,1} +\sum_{i=2}^N(-1)^{a_{n,1}}\tilde{A}_{1,0}\otimes A_{n,0}\right\rangle\leqslant  \beta_C
\end{eqnarray}
where,
\begin{eqnarray}\label{overAN}
    \tilde{A}_{1,0}=\frac{A_{0}-A_{1}}{\sqrt{2}},\qquad\tilde{A}_{1,1}=\frac{A_{0}+A_{1}}{\sqrt{2}}.
\end{eqnarray}
The classical bound of the above Bell inequalities is $\beta_C=\sqrt{2}(N-1)$ for any $a_1,b_1$. Furthermore, the  initial state $\rho_{N}$ should maximally violate the Bell inequality $\mathcal{B}_{0,0,\ldots,0}$.

\begin{fakt}The maximal quantum value of the Bell expression $\mathcal{B}_{\mathbf{a_1}}$ is $2(N-1)$ and it is achieved by the following observables 
\begin{eqnarray}\label{GHZObsN}
A_{1,0}=\frac{X+Z}{\sqrt{2}},\quad A_{1,1}= \frac{X-Z}{\sqrt{2}},\quad A_{n,0}=Z,\quad A_{n,1}=X, \quad\textrm{for }n=2,\ldots,N,
\end{eqnarray}
as well as the GHZ-like states 
\begin{equation}\label{GHZvecsN}
\ket{\phi_{\mathbf{a_1}}}=\frac{1}{\sqrt{2}}(\ket{a_{1,1}\ldots a_{N,1}}+(-1)^{a_{1,1}}\ket{a_{1,1}^{\perp}\ldots a_{N,1}^{\perp}}),
\end{equation}
where $a_{n,1}^{\perp}=1-a_n$. 
\end{fakt}
\begin{proof}
The following proof is based on Ref. \cite{sarkar2023universal}. Let us first write the Bell operator corresponding to each of the Bell expressions $\mathcal{I}_l$ \eqref{BE1NN} as
%
\begin{equation}\label{BE1Nop}
\hat{\mathcal{B}}_{\mathbf{a_1}}=(-1)^{a_{1,1}} (N-1)\tilde{A}_{1,1}\otimes\bigotimes_{n=2}^N A_{n,1} +\sum_{n=2}^N(-1)^{a_{1,1}+a_{n,1}}\tilde{A}_{1,0}\otimes A_{n,0},
\end{equation}
Now, for each $\hat{\mathcal{B}}_{\mathbf{a_1}}$ we can find the following sum-of-squares (SOS) decomposition,
\begin{eqnarray}\label{SOS1N}
   2\left[2(N-1)\,\I-\hat{\mathcal{B}}_{\mathbf{a_1}}\right]\geq(N-1)P_{a_{1,1}}^2+\sum_{i=2}^NQ_{a_{n,1},a_{1,1}}^2, 
\end{eqnarray}
where 

\begin{equation}\label{SOS2_1}
    P_{a_{1,1}}=(-1)^{a_{1,1}}\tilde{A}_{1,1}-\bigotimes_{n=2}^N A_{n,1},\qquad
   Q_{a_{n,1},a_{1,1}}=(-1)^{a_{1,1}+a_{n,1}}\tilde{A}_{1,0}-A_{n,0}
\end{equation}
Notice that expanding the terms on the right-hand side of the above decomposition \eqref{SOS1N} gives us
\begin{eqnarray}\label{SOS11N}
 (N-1)\left(\tilde{A}_{1,1}^2+\tilde{A}_{1,0}^2\right)+(N-1)\bigotimes_{n=2}^N A_{n,1}^2+\sum_{i=2}^NA_{n,0}^2-2\ \mathcal{B}_{\mathbf{a_1}}\leqslant \  2\left[2(N-1)\ \I-\ \mathcal{B}_{\mathbf{a_1}}\right]
\end{eqnarray}
where to arrive at the last line we used the fact that $\tilde{A}_{i,j}^2\leqslant \I$. Thus, $2(N-1)$ is an upper bound on $\mathcal{B}_{\mathbf{a_1}}$ when using quantum theory. One can easily verify that the Bell inequalities \eqref{BE1N} attain the value $2(N-1)$ when the states shared among the parties are the GHZ-like state $\ket{\phi_{\mathbf{a_1}}}$ \eqref{GHZvecsN} and the observables given in Eq. (\ref{GHZObsN}). Thus, the quantum bound of the Bell inequalities \eqref{BE1NN} for any $\mathbf{a_1}$ is $2(N-1)$.
\end{proof}

An important observation from the above SOS decomposition is that when the maximal violation of the above Bell inequalities is attained, the right-hand side of Eq. \eqref{SOS2_1} is $0$. Thus, any state $\ket{\psi}$ and observables $A_{i,j}$ that attain the maximal violation of the Bell inequalities \eqref{BE1N} for any $\mathbf{a_1}$ is given by
    \begin{equation}\label{SOSrel1N}
          (-1)^{a_{1,1}}\tilde{A}_{1,1}\ket{\psi}= \bigotimes_{n=2}^N A_{n,1}\ket{\psi},\qquad(-1)^{a_{1,1}+a_{n,1}}\tilde{A}_{1,0}\ket{\psi}=A_{n,0}\ket{\psi}\quad (n=2,\ldots,N),
    \end{equation}
Furthermore, one can also obtain from Eq. \eqref{SOS11N} that the observables that attain the maximum violation of the Bell inequalities \eqref{BE1N} must be unitary, that is, $A_{i,j}^2=\I$. These relations will be particularly useful for self-testing.

Along with the above statistics, one also needs to observe when $a_{n,1}=0$ for all $n$ and $x_{1,1}=x_{2,1}=1,\ x_{n,1}=0 (n=3,\ldots,N)$ that
\begin{eqnarray}\label{extrastatN}
    \langle \tilde{A}_{1,0}\otimes\I\rangle=-1,\qquad\langle\I\otimes A_{n,1}\rangle=1.\qquad n=2,\ldots,N.
\end{eqnarray}

\subsection{Model-independent certification}
\begin{thm}\label{theorem1N}
Assume that the Bell inequality $\mathcal{B}_{0,\ldots,0}$ \eqref{BE1NN} is maximally violated at $t_1$. Furthermore, when the parties obtain the outcome $\mathbf{a_1}$ with inputs $x_{1,1}=x_{2,1}=0$ and $x_{n,1}=1(n=3,\ldots,N)$, then the Bell inequalities $\mathcal{B}_{\mathbf{a_1}}$ \eqref{BE1NN} for any $\mathbf{a_1}$ are maximally violated at $t_2$. Along with it, when the parties observe the outcomes $a_{n,1}=0$ for all $n$ and $x_{1,1}=x_{2,1}=1,\ x_{n,1}=0 (n=3,\ldots,N)$, then the condition \eqref{extrastatN} is also satisfied. The source $P$ prepares the state $\rho_{N}$ acts on $\bigotimes_{n=1,\ldots,N}\mathcal{H}_{A_n(t_1)}$, the interaction between the particles is represented by the unitary $V_N:\bigotimes_{n=1,\ldots,N}\mathcal{H}_{A_n(t_1)}\rightarrow\bigotimes_{n=1,\ldots,N}\mathcal{H}_{A_n(t_2)}$ and the local observables of all parties are given by $A_{n,i}\ (i=0,1)$ that acts on $\mathcal{H}_{A_n}$ respectively where $\mathcal{H}_{A_n}\equiv\mathcal{H}_{A_n(t_1)}\cup\mathcal{H}_{A_n(t_2)}$. Then, there exist local unitary transformations $U_{n(t_i)}:\mathcal{H}_{A_n(t_i)}\rightarrow(\mathbb{C}^2)_{A'_n(t_i)}\otimes\mathcal{H}_{A''_n(t_i)}$ for $i=1,2$ 
 such that:
\begin{enumerate}
    \item  The state $\rho_{N}$ is certified as
\begin{eqnarray}\label{A10N}
\bigotimes_nU_{n(t_1)}\ \rho_{N}\ \bigotimes_nU_{n(t_1)}^{\dagger}=\proj{\phi_{0,\ldots,0}}_{A'_1(t_1)\ldots A_N'(t_1)}\otimes\xi_{A_1''(t_1)\ldots A_N''(t_1)}.
\end{eqnarray}
\item The observables of Alice and Bob are certified as  
\begin{eqnarray}\label{mea1N}
&U_{1(t_i)}\overline{A}_{1,0(t_i)}\,U_{1}^{\dagger}=\left(\frac{X+Z}{\sqrt{2}}\right)_{A_1'(t_i)}\otimes\I_{A_1''(t_i)}, \qquad U_{1(t_i)}\overline{A}_{1,1}\,U_{1(t_i)}^{\dagger}=\left(\frac{X-Z}{\sqrt{2}}\right)_{A_1'(t_i)}\otimes\I_{A_1''(t_i)},\nonumber\\
&U_{n(t_i)}\overline{A}_{n,0(t_i)}\,U_{n(t_i)}^{\dagger}=Z_{A_n'(t_i)}\otimes\I_{A_n''(t_i)},\qquad U_{n(t_i)}\overline{A}_{n,1(t_i)}\,U_{n(t_i)}^{\dagger}=X_{A_n'(t_i)}\otimes\I_{A_n''(t_i)}
\end{eqnarray}
for any $n=1,\ldots,N$ and $i=1,2,\ j=0,1$ where $\overline{A}_{n,j(t_i)}=\Pi_{i}^nA_{n,j}\Pi_{i}^n$ such that $\Pi_{i}^n$ denotes the projection onto the Hilbert space $\mathcal{H}_{A_n(t_i)}$.

\item The unitary $V_N$ is certified as $\bigotimes_nU_{n(t_2)}\ V_N\ \bigotimes_nU_{n(t_1)}^{\dagger}=\mathcal{U_N}\otimes V_{\mathrm{aux}}$ where $\mathcal{U_N}$ is given in Eq. \eqref{Nuni} and $V_{\mathrm{aux}}$ is a unitary that maps  $\bigotimes_{n}\mathcal{H}_{A_n''(t_1)}$ to $\bigotimes_{n}\mathcal{H}_{A_n''(t_2)}$.
\end{enumerate}
\end{thm}
\begin{proof}
The proof is divided into two major parts. In the first part, we certify the local observables of Alice and Bob along with the states prepared by the source $P$ and the post-interaction states. Then in the second part, using the certified states and observables we certify the unitary $V_N$.
    \begin{center}
        \textbf{Local observables and the state prepared by the source $P$}
    \end{center}
Consider the $N-$party generalising of the scenario presented in Fig. 1 of the manuscript at time $t_1$, that is, before the $N$ systems interact. Now, the quantum state $\rho_{N}$ acts on $\bigotimes_{n=1,\ldots,N}\mathcal{H}_{A_n(t_1)}$ and thus we consider the projection of the observables $A_{n,j}$ onto the support of the local states as $\overline{A}_{n,j(t_1)}$ where
\begin{eqnarray}
    \overline{A}_{n,j(t_1)}=\Pi_{1}^nA_{n,j}\Pi_{1}^n\qquad (j=0,1)
\end{eqnarray}
such that $\Pi_{1}^n$ denote the projectors onto $\mathcal{H}_{A_n(t_1)}$.

Following \cite{sarkar2023universal} and $a_{n,1}=0$ for any $n$, we obtain from Eqs. \eqref{SOSrel1N} that
\begin{eqnarray}
    \left\{\overline{{A}}_{0,1},\overline{{A}}_{1,1}\right\}=0,\qquad \left\{\overline{{B}}_{0,1},\overline{{B}}_{1,1}\right\}=0.
\end{eqnarray}
Consequently, there exists a unitary transformation $U_{n(t_1)}:\mathcal{H}_{n(t_1)}\rightarrow(\mathbb{C}^2)_{n'(t_1)}\otimes\mathcal{H}_{n''(t_1)}$ to obtain \eqref{mea1N} for $i=1$. Utilising now the relation \eqref{SOSrel1N} for $a_1=b_1=0$ and following exactly from \cite{sarkar2023universal}, we can self-test the state $\rho_N$ as
\begin{eqnarray}
   \bigotimes_nU_{n(t_1)}\rho_{N} \bigotimes_nU_{n(t_1)}^{\dagger}=\proj{\phi_{0,\ldots,0}}_{A_1'(t_1)\ldots A'_N(t_1)}\otimes\xi_{A_1''(t_1)\ldots A''_N(t_1)}.
\end{eqnarray}

Now the post-measurement states, $\rho'_{\mathbf{a_1},\mathbf{x_1}}$, of all parties undergo a joint unitary transformation $V_N:\bigotimes_{n}\mathcal{H}_{A_n(t_1)}\rightarrow\bigotimes_{n}\mathcal{H}_{A_n(t_2)}$ to get the post-interaction states, denoted by $\sigma_{\mathbf{a_1},\mathbf{x_1}}$.
As the measurements and states at time $t_1$ are self-tested, 
the post-measurement states, $\rho'_{\mathbf{a_1},\mathbf{x_1}}$ are given by
\begin{eqnarray}\label{pre-int stateN}
\bigotimes_nU_{n(t_1)}\rho'_{\mathbf{a_1},\mathbf{x_1}}\bigotimes_nU_{n(t_1)}^{\dagger}=\bigotimes_n\proj{e_{a_{n,1},x_{n,1}}}_{A'_n(t_1)}\otimes\xi_{A_1''(t_1)\ldots A''_N(t_1)}
\end{eqnarray}
where $\ket{e_{a_{n,1},x_{n,1}}}$ is the eigenvector corresponding to the $a_{n,1}$ outcome of $x_{n,1}$ target measurement.

Consider again the generalised scenario of Fig. 1 of the manuscript at time $t_2$, that is, after the interaction.
The post-interaction states $\sigma_{\mathbf{a_1},0,0,1,\ldots,1}$ maximally violate the Bell inequalities \eqref{BE1N} for any $\mathbf{a_1}$. Consequently, they are also certified from the maximal violation of the Bell inequalities \eqref{BE1NN}. For this purpose, we again utilise the relation \eqref{SOSrel1} for any $\mathbf{a_1}$ and following exactly from \cite{sarkar2023universal}, we first obtain that
\begin{eqnarray}
    \left\{\mathbb{A}_{n,0}^{(\mathbf{a_1})},\mathbb{A}_{n,1}^{(\mathbf{a_1})}\right\}=0\qquad \forall \mathbf{a_1}
\end{eqnarray}
where $\mathbb{A}_{n,j}^{(\mathbf{a_1})}=\Pi_{n}^{(\mathbf{a_1})}A_{n,j}\Pi_{n}^{(\mathbf{a_1})}$  for $j=0,1$ such that $\Pi_{n}^{(\mathbf{a_1})}$ is the projection onto the local supports of the states. 
Now, utilizing Lemma 2 of \cite{sarkar2023self}, we obtain that the local observables acting on the support of all states, the corresponding projector given by $\Pi_{2}^n=\cup_{\mathbf{a_1}}\Pi_{n}^{(\mathbf{a_1})}$, also anti-commute 
\begin{eqnarray}
    \{\overline{A}_{n,0(t_2)},\overline{A}_{n,1(t_2)}\}=0
\end{eqnarray}
such that $\overline{A}_{n,j(t_2)}=\Pi_{2}^nA_{n,j}\Pi_{2}^n$. Consequently, one can straightaway conclude that there exist unitaries $U_{n(t_2)}:\mathcal{H}_{A_n(t_2)}\rightarrow(\mathbb{C}^2)_{A_n'(t_2)}\otimes\mathcal{H}_{A_n''(t_2)}$ to obtain \eqref{mea1N} for $i=2$. Furthermore from \cite{sarkar2023universal}, the post-interaction states are certified using the relations \eqref{SOSrel1N} for all $\mathbf{a_1}$ using the certified observables \eqref{mea1N} at time $t_2$ as
\begin{eqnarray}\label{post-int stateN}
   \bigotimes_nU_{n(t_2)}\sigma_{\mathbf{a_1},0,0,1,\ldots,1}\bigotimes_nU_{n(t_2)}^{\dagger}=\proj{\phi_{\mathbf{a_1}}}_{A_1'(t_2)\ldots,A_N'(t_2)}\otimes\xi^{\mathbf{a_1}}_{A''_1(t_2)\ldots A''_n(t_2)}
\end{eqnarray}
where $\ket{\phi_{\mathbf{a_1}}}$ are given in \eqref{GHZvecsN}. Let us now utilize the certified pre-interaction states \eqref{pre-int stateN} and post-interaction ones \eqref{post-int stateN} to self-test the unitary $V_N$.
\newpage
   \begin{center}
        \textbf{The unitary $V_N$}
    \end{center}
As shown above, the Hilbert spaces $\mathcal{H}_{A_n(t_i)}$ decompose as $(\mathbb{C}^2)_{A_n'(t_i)}\otimes\mathcal{H}_{A_n''(t_i)}$ for any $n,i$. Thus, we can express $V_N$ as

\begin{eqnarray}\label{VgenN}
   \bigotimes_nU_{n(t_2)}\ V_N\  \bigotimes_nU_{n(t_2)}^{\dagger}=\sum_{i_1,\ldots,i_N,j_1,\ldots,j_N=0,1}\ket{\overline{i}_1i_2i_3^x\ldots i_{N}^x}\bra{\overline{j}_{1}j_2j_3^x\ldots j_{N}^x}\otimes V_{i_1,\ldots,i_N,j_1,\ldots,j_N}
\end{eqnarray}
where $\ket{\overline{i}}$ are the eigenvectors of $(X+Z)/\sqrt{2}$ as given in \eqref{X+Z}. Furthermore, the state $\ket{\overline{i_1}\ldots i_{N}^x}\in\bigotimes_{n}\mathbb{C}^2_{A_n'(t_1)}$ and $\ket{\overline{j_1}\ldots j_{N}^x}\in\bigotimes_{n}\mathbb{C}^2_{A_n'(t_2)}$ along with $V_{i_1,\ldots,i_N,j,\ldots,j_N}$ maps  $\bigotimes_{n}\mathcal{H}_{A_n''(t_1)}$ to $\bigotimes_{n}\mathcal{H}_{A_n''(t_2)}$.

Now, suppose when $x_{n,1}=a_{n,1}=0$ for any $n$, then as described above the Bell inequality $\mathcal{B}_{0,\ldots,0}$ is maximally violated. Thus from the relation $V_N\rho'_{\mathbf{a_1},\mathbf{x_1}}V_N^{\dagger}=\sigma_{\mathbf{a_1},\mathbf{x_1}}$, we can purify the states on the left and right sides by adding ancillary systems $E$ to obtain that 
\begin{eqnarray}
V_N\otimes\I_E\ket{\rho'_{0,\ldots,0}}_{A_1(t_1)\ldots A_N(t_1)E}=\ket{\sigma_{0,\ldots,0}}_{A_1(t_2)\ldots A_N(t_2)E}
\end{eqnarray}
such that $\Tr_E (\proj{\rho'_{0,\ldots,0}})=\rho'_{0,\ldots,0}$ and $\Tr_E (\proj{\sigma_{0,\ldots,0}})=\sigma_{0,\ldots,0}$.
Now, using the certified states $\rho'_{0,\ldots,0}$ and $\sigma_{0,\ldots,0}$ from Eqs. \eqref{pre-int stateN} and \eqref{post-int stateN} respectively and then expanding $V_N$ using \eqref{VgenN}, we obtain
\begin{eqnarray}\label{V00}
    \sum_{i_1,\ldots,i_N=0,1}\ket{\overline{i}_1i_2i_3^x\ldots i_{N}^x} \left(V_{i_1,\ldots,i_N,0,\ldots,0}\otimes\I_E\ket{\xi}\right)=\ket{\phi_{0\ldots0}}\otimes\ket{\xi^{0\ldots0}}.
\end{eqnarray}
Multiplying by $\bra{\overline{i_1}i_2i_3^x\ldots i_{N}^x}$ for all $i_1,\ldots,i_N$ on both sides of the above formula gives us the following four conditions
\begin{eqnarray}
2^{\frac{N-1}{2}}\ V_{0,1,i_3,\ldots,i_N,0,\ldots,0}\otimes\I_E\ket{\xi}&=&(-1)^{i_3+\ldots+i_N}\sin(\pi/8)\ket{\xi^{0,0}},\nonumber\\ 2^{\frac{N-1}{2}}\ V_{1,1,i_3,\ldots,i_N,0,\ldots,0}\otimes\I_E\ket{\xi}&=&(-1)^{i_3+\ldots+i_N}\cos(\pi/8)\ket{\xi^{0,0}},\nonumber\\
2^{\frac{N-1}{2}}\ V_{0,0,i_3,\ldots,i_N,0,\ldots,0}\otimes\I_E\ket{\xi}&=&\cos(\pi/8)\ket{\xi^{0,0}},\qquad2^{\frac{N-1}{2}}\ V_{1,0,i_3,\ldots,i_N,0,\ldots,0}\otimes\I_E\ket{\xi}=-\sin(\pi/8)\ket{\xi^{0,0}}
\end{eqnarray}
for all $i_3,\ldots,i_N$.
Thus, from the above conditions we can conclude that
\begin{eqnarray}
    \frac{(-1)^{i_3+\ldots+i_N}}{\sin(\pi/8)}V_{0,1,i_3,\ldots,i_N,0,\ldots,0}\otimes\I_E\ket{\xi}&=&\frac{(-1)^{i_3+\ldots+i_N}}{\cos(\pi/8)}V_{1,1,i_3,\ldots,i_N,0,\ldots,0}\otimes\I_E\ket{\xi}\nonumber\\&=& \frac{1}{\cos(\pi/8)}V_{0,0,i_3,\ldots,i_N,0,\ldots,0}\otimes\I_E\ket{\xi}=\frac{-1}{\sin(\pi/8)}V_{1,0,i_3,\ldots,i_N,0,\ldots,0}\otimes\I_E\ket{\xi}.\qquad
\end{eqnarray}
Now, without loss of generality one can consider that $\xi$ is full-rank as the unitary $V$ can be characterised only on the support of the states $\rho'$. Thus, from the above formula taking partial trace over $E$ gives us
\begin{eqnarray}
     \frac{2^{\frac{N-1}{2}}(-1)^{i_3+\ldots+i_N}}{\sin(\pi/8)}V_{0,1,i_3,\ldots,i_N,0,\ldots,0}&=&\frac{2^{\frac{N-1}{2}}(-1)^{i_3+\ldots+i_N}}{\cos(\pi/8)}V_{1,1,i_3,\ldots,i_N,0,\ldots,0}\nonumber\\&=& \frac{2^{\frac{N-1}{2}}}{\cos(\pi/8)}V_{0,0,i_3,\ldots,i_N,0,\ldots,0}=\frac{-2^{\frac{N-1}{2}}}{\sin(\pi/8)}V_{1,0,i_3,\ldots,i_N,0,\ldots,0}=V_{0,\ldots,0}.
\end{eqnarray}

Similarly, considering the other outputs $a_{n,1}$ and $x_{1,1}=x_{2,1}=0$ and $x_{n,1}=1(n=3,\ldots,N)$ for any $n$ and recalling that 
\begin{eqnarray}\label{36N}
V_N\otimes\I_E\ket{\rho'_{\mathbf{a_1},0,0,1\ldots,1}}_{A(t_1)B(t_1)E}=\ket{\sigma_{\mathbf{a_1},0,0,1\ldots,1}}_{A(t_2)B(t_2)E}
\end{eqnarray}
where states $\ket{\rho'_{\mathbf{a_1},0,0,1\ldots,1}}$ and $\ket{\sigma_{\mathbf{a_1},0,0,1\ldots,1}}$ are certified in Eqs. \eqref{pre-int stateN} and \eqref{post-int stateN} respectively. Consequently, from the general expression of the unitary $V_N$ \eqref{VgenN} and the above condition \eqref{36N}, we obtain that
\begin{eqnarray}\label{VabN}
     \sum_{i_1,\ldots,i_N=0,1}\ket{\overline{i}_1i_2i_3^x\ldots i_{N}^x}\left(V_{i_1,\ldots,i_N,\mathbf{a_1}}\otimes\I_E\ket{\xi}\right)=\ket{\phi_{\mathbf{a_1}}}\otimes\ket{\xi^{\mathbf{a_1}}}.
\end{eqnarray}
Multiplying by $\bra{\overline{i_1}i_2i_3^x\ldots i_{N}^x}$ for all $i_1,\ldots,i_N$ on both sides of the above formula gives us the following conditions for any $\mathbf{a_1}$
\begin{eqnarray}
V_{i_1,\ldots,i_N,\mathbf{a_1}}\otimes\I_E\ket{\xi}&=&\bra{\overline{i}_1i_2i_3^x\ldots i_{N}^x}\phi_{\mathbf{a_1}}\rangle\ket{\xi^{\mathbf{a_1}}},\qquad i_1,\ldots,i_N=0,1
\end{eqnarray}
Thus, from the above conditions and again using the fact that $\xi$ is full-rank, we can conclude that
\begin{eqnarray}
      \frac{V_{i_1,\ldots,i_N,\mathbf{a_1}}}{\bra{\overline{i}_1i_2i_3^x\ldots i_{N}^x}\phi_{\mathbf{a_1}}\rangle}=V_{\mathbf{a_1}}\qquad i_1,\ldots,i_N=0,1.
\end{eqnarray}
Now, using the above condition in Eq. \eqref{Vgen} gives us 
\begin{eqnarray}\label{Vgen1N}
  \bigotimes_nU_{n(t_2)}\ V_N\  \bigotimes_nU_{n(t_2)}^{\dagger}&=&\sum_{i_1,\ldots,i_N,\mathbf{a_1}=0,1}\ket{\overline{i}_1i_2i_3^x\ldots i_{N}^x}\!\bra{\overline{a}_{1,1}a_{2,1}a_{3,1}^x\ldots a_{N,1}^x}\bra{\overline{i}_1i_2i_3^x\ldots i_{N}^x}\phi_{\mathbf{a_1}}\rangle\otimes V_{\mathbf{a_1}}\nonumber\\
  &=&\sum_{i_1,\ldots,i_N,\mathbf{a_1}=0,1}\ket{\overline{i}_1i_2i_3^x\ldots i_{N}^x}\!\bra{\overline{i}_1i_2i_3^x\ldots i_{N}^x}\phi_{\mathbf{a_1}}\rangle\!\bra{\overline{a}_{1,1}a_{2,1}a_{3,1}^x\ldots a_{N,1}^x}\otimes V_{\mathbf{a_1}}\nonumber\\
   &=&\sum_{\mathbf{a_1}=0,1}\ket{\phi_{\mathbf{a_1}}}\!\bra{\overline{a}_{1,1}a_{2,1}a_{3,1}^x\ldots a_{N,1}^x}\otimes V_{\mathbf{a_1}}
\end{eqnarray}
where we used the fact that$\sum_{i_1,\ldots,i_N}\ket{\overline{i}_1i_2i_3^x\ldots i_{N}^x}\!\bra{\overline{i}_1i_2i_3^x\ldots i_{N}^x}=\I$.


Now, we finally utilize the fact that when $a_{n,1}=0$ for all $n$ and $x_{1,1}=x_{2,1}=1,\ x_{n,1}=0 (n=3,\ldots,N)$, then one satisfies the relation \eqref{extrastatN}. Let us exploit this relation to show that all $V_{\mathbf{a_1}}$ are equal. The state $\sigma_{0,\ldots,0,1,1,0\ldots,0}$ is given by
\begin{eqnarray}
 \bigotimes_nU_{n(t_2)}  \ket{\sigma_{0,\ldots,0,1,1,0\ldots,0}}= \left( \bigotimes_nU_{n(t_2)} V_N\otimes\I_E  \bigotimes_nU_{n(t_2)}^{\dagger}\right) \bigotimes_nU_{A_n(t_1)}\ket{\rho'_{0,\ldots,0,1,1,0\ldots,0}}.
\end{eqnarray}
Using the fact that $\rho'_{0,\ldots,0,1,1,0\ldots,0}$ is certified as in \eqref{pre-int stateN} and the form of $V_N$ from \eqref{Vgen1N}, we obtain from the above formula that
\begin{eqnarray}\label{83}
    \bigotimes_nU_{n(t_2)}  \ket{\sigma_{0,\ldots,0,1,1,0\ldots,0}}=\frac{1}{2^\frac{N}{2}}\left(\sum_{\mathbf{a_1}=0,1}\ket{\phi_{\mathbf{a_1}}}\otimes V_{\mathbf{a_1}}\right)\otimes\I_E\ket{\xi}.
\end{eqnarray}
Consider now the relation $\langle \I\otimes A_{n,1}\rangle=1$, which using Cauchy-Schwarz inequality can be written as 
\begin{eqnarray}\label{84}
    \I\otimes A_{n,1}\ket{\sigma_{0,\ldots,0,1,1,0\ldots,0}}=\ket{\sigma_{0,\ldots,0,1,1,0\ldots,0}}.
\end{eqnarray}
Notice that $\sigma_{0,\ldots,0,1,1,0\ldots,0}$ acts on $\bigotimes_n\mathcal{H}_{A_n(t_2)}$ and $A_{n,1}$ is certified on it as in Eq. \eqref{mea1N}. Let us now observe that for $n=2,\ldots,N$
\begin{eqnarray}
   \I\otimes X_{A_n}\ket{\phi_{l_1\ldots,l_n,\ldots, l_N}}=\ket{\phi_{l_1\ldots,l^{\perp}_n,\ldots, l_N}}
\end{eqnarray}
where $l^{\perp}_n=1-l_n$.
Plugging in the state \eqref{83} into the above condition \eqref{84} for any $n$, we obtain that
\begin{eqnarray}
 \left(\sum_{a_{1,1},\ldots a_{N,1}=0,1}\ket{\phi_{a_{1,1}\ldots,a^{\perp}_{n,1},\ldots, a_{N,1}}}\otimes V_{a_{1,1},\ldots,a_{n,1},\ldots a_{N,1}}\right)\otimes\I_E\ket{\xi}= \left(\sum_{a_{1,1},\ldots a_{N,1}=0,1}\ket{\phi_{a_{1,1},\ldots,a_{n,1},\ldots a_{N,1}}}\otimes V_{a_{1,1},\ldots,a_{n,1},\ldots, a_{N,1}}\right)\otimes\I_E\ket{\xi}\nonumber\\
\end{eqnarray}
which by utilising the fact that $\xi$ is full-rank allows us to conclude that
\begin{eqnarray}\label{Vgen2N}
   V_{a_{1,1},\ldots,a^{\perp}_{n,1},\ldots a_{N,1}}=V_{a_{1,1},\ldots,a_{n,1},\ldots a_{N,1}}\qquad n=2,\ldots,N.
\end{eqnarray}
Now, utilising the relation $\langle \tilde{A}_{1,0}\otimes \I\rangle=-1$, which again using Cauchy-Schwarz inequality can be written as 
\begin{eqnarray}\label{90}
\tilde{A}_{1,0}\otimes\I\ket{\sigma_{0,\ldots,0,1,1,0\ldots,0}}=-\ket{\sigma_{0,\ldots,0,1,1,0\ldots,0}}
\end{eqnarray}
Notice now that
\begin{eqnarray}
  Z_{A_n}\otimes\I\ket{\phi_{l_1,l_2\ldots, l_N}}=-\ket{\phi_{l^{\perp}_1,l^{\perp}_2,\ldots, l^{\perp}_N}}
\end{eqnarray}
using which we can obtain from Eqs. \eqref{90} and \eqref{83}  that
\begin{eqnarray}
  \left(\sum_{a_{1,1},\ldots a_{N,1}=0,1}\ket{\phi_{a_{1,1}^{\perp},a^{\perp}_{2,1},\ldots, a_{N,1}^{\perp}}}\otimes V_{a_{1,1},a_{2,1},\ldots a_{N,1}}\right)\otimes\I_E\ket{\xi}= \left(\sum_{a_{1,1},\ldots a_{N,1}=0,1}\ket{\phi_{a_{1,1},\ldots,a_{n,1},\ldots a_{N,1}}}\otimes V_{a_{1,1},a_{2,1},\ldots, a_{N,1}}\right)\otimes\I_E\ket{\xi}.\nonumber\\
\end{eqnarray}
This allows to conclude that
\begin{eqnarray}\label{Vgen3N}
    V_{a_{1,1}^{\perp},a_{2,1}^{\perp},\ldots a_{N,1}^{\perp}}=V_{a_{1,1},a_{2,1},\ldots a_{N,1}}.
\end{eqnarray}
It is now simple to observe from Eqs. \eqref{Vgen2N} and \eqref{Vgen3N} that  $V_{\mathbf{a_1}}=V_{\mathrm{aux}}$ for any $\mathbf{a_1}$. Consequently, one can straightforwardly conclude from \eqref{Vgen1N} that
\begin{eqnarray}
 \bigotimes_nU_{n(t_2)}\ V_N\  \bigotimes_nU_{n(t_2)}^{\dagger}=\sum_{\mathbf{a_1}=0,1}\ket{\phi_{\mathbf{a_1}}}\!\bra{\overline{a}_{1,1}a_{2,1}a_{3,1}^x\ldots a_{N,1}^x}\otimes V_{\mathrm{aux}}.
\end{eqnarray}
\end{proof}

\end{document}